\documentclass[12pt]{amsart}

\usepackage{amsmath,amssymb,amsthm,amstext,amsfonts}
\usepackage{physics}
\usepackage{mathtools}
\usepackage{hyperref}
\usepackage{graphicx}
\usepackage{bm}
\usepackage{geometry}
\usepackage{enumerate}
\usepackage{xcolor}

\numberwithin{equation}{section}

\newtheorem{theorem}{Theorem}[section]
\newtheorem{corollary}[theorem]{Corollary}
\newtheorem{proposition}[theorem]{Proposition}
\newtheorem{lemma}[theorem]{Lemma}
\newtheorem*{theorem*}{Theorem}

\theoremstyle{definition}
\newtheorem{definition}[theorem]{Definition}
\newtheorem{example}[theorem]{Example}

\theoremstyle{remark}
\newtheorem{remark}[theorem]{Remark}

\begin{document}

\title[Hamiltonian Simulation]{Complexity Bounds for Hamiltonian Simulation in Unitary Representations}

\author{Naihuan Jing}
\address{Department of Mathematics,
   North Carolina State University,
   Raleigh, NC 27695, USA}
\email{jing@ncsu.edu}
\author{Molena Nguyen}
\address{Department of Mathematics,
   North Carolina State University,
   Raleigh, NC 27695, USA}
\email{thnguy22@ncsu.edu}

\thanks{*
Partially supported by Simons Foundation (grant no. MP-TSM-00002518)\\*Corresponding author: Naihuan Jing.}

\date{November 18, 2025}


\begin{abstract}
We develop a representation-theoretic framework for Hamiltonian simulation on a compact semisimple Lie group $G$. For any unitary representation $\rho$
on a finite-dimensional Hilbert space \(V\) with differential \(d\rho : \mathfrak{g} \to \mathfrak{u}(V)\) for the Lie algebra $\mathfrak g$, we consider the Hamiltonian evolution
\[
U_X(t) \coloneqq \rho(\exp(tX)) = e^{t\,d\rho(X)}, \qquad t\in\mathbb{R}.
\]
For any $X\in\mathfrak{g}$ and its complexification $
X_\mathbb{C} = X_0 + \sum\limits_{\alpha\in\Delta} x_\alpha E_\alpha
$ associated with the root system $\Delta$,
we introduce the numerical invariants \emph{root activity} and \emph{root curvature} functionals
\begin{align*}
\mathcal{A}_p(X)
&\coloneqq
\Bigl(\sum_{\alpha\in\Delta} |x_\alpha|^p \,\|d\rho(E_\alpha)\|_{\mathrm{op}}^p\Bigr)^{1/p},
\quad 1\le p<\infty\\
\mathcal{C}(X)
&\coloneqq
\Bigl(\sum_{\alpha\in\Delta}
|\alpha(X_0)|^2\,|x_\alpha|^2 \,\|d\rho(E_\alpha)\|_{\mathrm{op}}^2\Bigr)^{1/2},
\end{align*}
where \(\|\cdot\|_{\mathrm{op}}\) is the operator norm on \(\mathrm{End}(V)\).
We first describe
how the Hamiltonian \(d\rho(X)\) is distributed along the directions of root spaces $\mathfrak{g}_\alpha$.

Our main analytic result provides a local error bound for the symmetric torus--root splitting, that is, for the approximation
\[
e^{t(d\rho(X_0)+d\rho(X_{\mathrm{root}}))}
\approx
e^{\frac{t}{2}d\rho(X_0)} e^{t d\rho(X_{\mathrm{root}})} e^{\frac{t}{2}d\rho(X_0)}.\]
Specifically, we show that for each fixed \(X\in\mathfrak{g}\) there exists a constant \(C_X>0\) such that
\[
\bigl\|
e^{t(d\rho(X_0)+d\rho(X_{\mathrm{root}}))}
-
e^{\frac{t}{2}d\rho(X_0)} e^{t d\rho(X_{\mathrm{root}})} e^{\frac{t}{2}d\rho(X_0)}
\bigr\|_{\mathrm{op}}
\le C_X\,t^{3}\,\bigl(\mathcal{C}(X)+\mathcal{A}_1(X_{\mathrm{root}})\bigr)
\]
for all sufficiently small \(|t|\), where \(C_X\) depends on \(X\) and on \((\mathfrak{g},\rho)\), but not on \(\dim V\) beyond its appearance in \(\rho\).  Thus the leading constant in the \(t^3-\)term is expressed directly in terms of the root data \(\{x_\alpha\}\), \(\{\alpha(X_0)\}\), and \(\{d\rho(E_\alpha)\}\), rather than global norms of \(X\) or nested commutators, yielding refined Trotter--Suzuki error bounds that reflect the root-space structure.
We also introduce a root-gate circuit model whose elementary gates are one-parameter exponentials \(\exp(s\,d\rho(H))\) for \(H\in\mathfrak{t}\) and \(\exp(s\,d\rho(E_\alpha))\) for \(\alpha\in\Delta\), with parameter \(s\in\mathbb{R}\), the coefficients \(x_\alpha\) and \(\mathcal{A}_1(X_{\mathrm{root}})\) quantify the flow along the root directions and suggest a bound \(t\,\mathcal{A}_1(X_{\mathrm{root}})\) for simulation depth.  Finally, we test this on spin$-$chain Hamiltonians on \((\mathbb{C}^2)^{\otimes n}\subset\mathfrak{su}(2^n)\), where root spaces are spanned by matrix units, \(\mathcal{A}_p\), and \(\mathcal{C}\), which gives sharper complexity bounds and dimension$-$free representation$-$theoretic invariants.

\end{abstract}

\maketitle

\tableofcontents

\section{Introduction}
\label{sec:intro}

Hamiltonian simulation lies at the interface of quantum computation, numerical analysis, and Lie theory.  Let \(\mathcal{H}\) be a finite-dimensional complex Hilbert space and let \(H\) be a time-independent self-adjoint operator on \(\mathcal{H}\) (the Hamiltonian).  The associated unitary time evolution is
\[
U(t) \coloneqq e^{-i t H}, \qquad t \in \mathbb{R}.
\]
A basic algorithmic task is to approximate \(U(t)\), up to a prescribed accuracy, using a fixed finite family of elementary unitary operators on \(\mathcal{H}\) (often called \emph{gates}).

The modern theory of Hamiltonian simulation goes back to Lloyd's universality theorem \cite{lloyd1996universal}, which shows that local Hamiltonian dynamics generated by sums of terms acting on only a few degrees of freedom can be approximated by finite products of short-time evolutions generated by those local terms.  Since then, a large body of work has developed algorithmic methods for simulating many-body and quantum-chemistry Hamiltonians, beginning with schemes based on Lie--Trotter and Strang product formulas \cite{trotter1959product,berry2007efficient,aspuruguzik2005,georgescu2014quantum,whitfield2011,kassal2011,cao2019quantumchem} and later including approaches based on Taylor-series expansions, qubitization, and quantum signal processing \cite{berry2015simulating,childs2018toward,low2017optimal,low2019hamiltonian,babbush2018encoding,childs2012hamiltonian}.  These methods give asymptotically optimal bounds on the number of required gates as a function of the evolution time \(t\) and the target precision, and they underlie many proposed quantum algorithms in chemistry, condensed-matter physics, and lattice models; see, for example, the surveys \cite{georgescu2014quantum,cao2019quantumchem}.

Many of the known algorithms are built from \emph{product formulas}.  In their simplest form these are the Lie--Trotter and Strang splittings, together with higher-order Suzuki refinements \cite{suzuki1990fractal,blanes2009magnus,blanescasas2016}, which approximate the exponential of a sum of operators by products of exponentials of the individual summands.  The associated error bounds are usually expressed in terms of operator norms of the Hamiltonian \(H\) and of its nested commutators, and can be derived from the Baker--Campbell--Hausdorff (BCH) and Magnus expansions \cite{iserles2000lie,hairer2006geometric,lubich2008,hochbruckostermann2010,magnus1954}.  Recent work has sharpened these Trotter error estimates for local lattice Hamiltonians and sparse systems \cite{childs2021theory,haah2018product,campbell2019random}, and has clarified how higher-order commutators contribute to the cost of simulation.  From the numerical-analysis point of view, these developments place Hamiltonian simulation within the general theory of geometric integrators for differential equations on Lie groups and manifolds \cite{blanes2009magnus,hairer2006geometric,lubich2008,sanzsernacalvo1994}.

Beyond such norm-based bounds, there is growing interest in understanding Hamiltonian simulation from a more \emph{geometric} and \emph{representation-theoretic} point of view.  In many applications the Hamiltonian \(H\) arises as
\[
H = d\rho(X)
\]
for some element \(X\) of a compact semisimple Lie algebra \(\mathfrak{g}\), where \(\rho\) is a unitary representation of the corresponding Lie group.  In this setting the additional symmetry suggests decompositions of \(X\) adapted to a Cartan subalgebra and the associated root-space decomposition, and it links simulation questions to geometric control theory on Lie groups.  For two-qubit and few-body systems, Cartan-type decompositions of \(\mathrm{SU}(4)\) and related groups have been used to classify nonlocal gates and to analyze time-optimal control \cite{khaneja2001time,khaneja2001cartan,zhang2003geometric,nielsen2006geometry}, and similar ideas appear in the study of controllability for spin systems \cite{khaneja2001time,albertinidalessandro2003,dalessandro2007}.

In another direction, K\"ok\c{c}u \emph{et al.}\ introduced a fixed-depth framework for Hamiltonian simulation based on Cartan decompositions \cite{kokcu2022fixeddepth}, showing that certain structured Hamiltonians admit exact or approximate simulations by circuits whose depth is bounded independently of the dimension of \(\mathcal{H}\); their analysis relates algebraic features such as sparsity and commutation relations to the difficulty of simulating the dynamics.  Finally, a series of works by Moody Chu \cite{chu2022geometry,Chu2023LaxDynamics,Chu2025PreparingHamiltonians} proposed a geometric approach to matrix exponentials and factorizations, viewing unitary evolutions as curves on Lie groups and symmetric spaces; in that setting curvature and commutators appear naturally in error bounds and factorization schemes.  The present paper is motivated by these geometric perspectives and develops curvature-sensitive estimates for Hamiltonian simulation in a representation-theoretic framework.

\subsection{Representation-theoretic setting and notation}

Throughout the paper we fix the following data, which are also used in the abstract.

\begin{itemize}
\item \(G\) denotes a connected compact semisimple Lie group with Lie algebra \(\mathfrak{g}\).
\item \(V\) is a finite-dimensional complex Hilbert space, with inner product \(\langle \cdot,\cdot\rangle_V\) and associated norm \(\|\cdot\|_V\).
\item \(\rho : G \to U(V)\) is a unitary representation of \(G\) on \(V\).
\item The differential of \(\rho\) is the Lie algebra representation
\[
d\rho : \mathfrak{g} \longrightarrow \mathfrak{u}(V),
\]
where \(\mathfrak{u}(V)\subset EndV)\) is the Lie algebra of skew-Hermitian operators on \(V\),
\[
\mathfrak{u}(V) \coloneqq \{ A \in EndV) : A^\ast = -A\}.
\]
\end{itemize}

For each element \(X\in\mathfrak{g}\) we write
\begin{equation}\label{eq:Ux-def}
U_X(t) \coloneqq \rho(\exp(tX)) = e^{t\,d\rho(X)}, \qquad t\in\mathbb{R},
\end{equation}
for the corresponding unitary evolution.  In physical language, the self-adjoint operator
\begin{equation}\label{eq:Hx-def}
H_X \coloneqq -i\,d\rho(X)
\end{equation}
is the Hamiltonian, and \(U_X(t)\) is the Schr\"odinger evolution generated by \(H_X\); see, for example, \cite{georgescu2014quantum,nielsen2000,varadarajan1974}.

To measure errors in Hamiltonian simulation we use the operator norm
\begin{equation}\label{eq:op-norm}
\|A\|_{\mathrm{op}} \coloneqq \sup_{\|v\|_V=1} \|A v\|_V,
\qquad A\in End(V).
\end{equation}
When no confusion is possible we simply write \(\|\cdot\|\) for \(\|\cdot\|_{\mathrm{op}}\).

A central algorithmic problem considered in this paper is the following qualitative question.

\begin{definition}[Hamiltonian simulation problem]\label{def:HS-problem}
Let \(X\in\mathfrak{g}\) and \(t\in\mathbb{R}\), and let \(\varepsilon>0\).  The Hamiltonian simulation problem for the pair \((X,t)\) is to construct a product
\[
W = U_1 U_2 \cdots U_N, \qquad U_j \in \mathcal{G},
\]
from a prescribed gate set \(\mathcal{G}\subset U(V)\) such that
\[
\|W - U_X(t)\|_{\mathrm{op}} \le \varepsilon,
\]
while keeping the length \(N\) as small as possible.  The dependence of \(N\) on \(\varepsilon\), \(t\), and structural data associated with \(X\) is the \emph{simulation complexity} of \(U_X(t)\) in the gate model \(\mathcal{G}\).
\end{definition}

The particular gate sets we consider later are adapted to the root structure of \(\mathfrak{g}\) and to the representation \(\rho\).

\subsection{Overview}
\label{subsec:intro-overview}

We now describe the main constructions and results of the paper in the representation-theoretic setting above.  Throughout, we fix \(G\), \(\mathfrak{g}\), \(V\), and \(\rho\) as in the previous subsection. We first recall some standard notations from Lie theory.

\medskip

\noindent\textbf{Root-space decomposition and normalization.}
We equip \(\mathfrak{g}\) with an \(Ad(G)\)-invariant inner product, for instance the negative of the Killing form, and use it to identify \(\mathfrak{g}\) with its dual \(\mathfrak{g}^\ast\); see \cite{helgason1978differential,knapp2002lie,bourbaki2005,hilgertneeb2012,fultonharris1991}.  Fix a maximal torus
\[
T \subset G
\]
with Lie algebra \(\mathfrak{t} \subset \mathfrak{g}\).  Denote by
\[
\mathfrak{g}_\mathbb{C} \coloneqq \mathfrak{g} \otimes_\mathbb{R} \mathbb{C},
\qquad
\mathfrak{t}_\mathbb{C} \coloneqq \mathfrak{t} \otimes_\mathbb{R} \mathbb{C}
\]
the complexifications of \(\mathfrak{g}\) and \(\mathfrak{t}\), and by \(\mathfrak{t}_\mathbb{C}^\ast\) the complex dual of \(\mathfrak{t}_\mathbb{C}\).

The root system
\[
\Delta \subset \mathfrak{t}_\mathbb{C}^\ast
\]
of the pair \((\mathfrak{g}_\mathbb{C},\mathfrak{t}_\mathbb{C})\) is defined in the usual way: for each nonzero linear functional \(\alpha \in \mathfrak{t}_\mathbb{C}^\ast\) such that the corresponding root space
\[
\mathfrak{g}_\alpha
\coloneqq
\{ Y \in \mathfrak{g}_\mathbb{C} : [H,Y] = \alpha(H)\, Y \ \text{for all } H \in \mathfrak{t}_\mathbb{C} \}
\]
is nontrivial, we call \(\alpha\) a \emph{root} and include it in \(\Delta\).  The complexified Lie algebra admits the root-space decomposition
\begin{equation}\label{eq:root-space-decomp}
\mathfrak{g}_\mathbb{C}
=
\mathfrak{t}_\mathbb{C} \oplus \bigoplus_{\alpha \in \Delta} \mathfrak{g}_\alpha.
\end{equation}
For each \(\alpha \in \Delta\) we fix a nonzero \emph{root vector}
\[
E_\alpha \in \mathfrak{g}_\alpha,
\]
chosen so that the root spaces are pairwise orthogonal with respect to the complex-bilinear extension of our \(Ad(G)\)-invariant inner product on \(\mathfrak{g}\).  This choice of root vectors \(\{E_\alpha\}_{\alpha\in\Delta}\) is fixed throughout the paper.

\medskip

\noindent\textbf{Decomposition of Hamiltonian generators.}
Let \(X \in \mathfrak{g}\) be a fixed Hamiltonian generator, so that \(H_X = -i\,d\rho(X)\) is the corresponding Hamiltonian on \(V\).  Write \(X_\mathbb{C} \in \mathfrak{g}_\mathbb{C}\) for its complexification.  With respect to the decomposition \eqref{eq:root-space-decomp}, there exist unique coefficients
\[
X_0 \in \mathfrak{t}_\mathbb{C},
\qquad
x_\alpha \in \mathbb{C} \ \text{for each } \alpha \in \Delta
\]
such that
\begin{equation}\label{eq:root-decomp-intro}
X_\mathbb{C}
=
X_0 + X_{\mathrm{root},\mathbb{C}},
\qquad
X_{\mathrm{root},\mathbb{C}} \coloneqq \sum_{\alpha\in\Delta} x_\alpha E_\alpha.
\end{equation}
There is a unique real element in \(\mathfrak{t}\) whose complexification is \(X_0\); for simplicity we again denote this element by \(X_0\).  We then define the \emph{root component} of \(X\) by
\[
X_{\mathrm{root}} \coloneqq X - X_0 \in \mathfrak{g},
\]
so that the complexification of \(X_{\mathrm{root}}\) is \(X_{\mathrm{root},\mathbb{C}}\).  The pair \((X_0,X_{\mathrm{root}})\) is the toral--root decomposition of the generator \(X\) with respect to the chosen Cartan subalgebra \(\mathfrak{t}\).  This decomposition, formulated rigorously in Section~\ref{sec:root-profiles}, agrees with the one used in the abstract.

Applying the differential representation, we obtain operators
\[
d\rho(X_0), \quad d\rho(X_{\mathrm{root}}) \in \mathfrak{u}(V),
\qquad
d\rho(E_\alpha) \in EndV) \ \text{for } \alpha \in \Delta,
\]
which we measure in the operator norm \(\|\cdot\|_{\mathrm{op}}\) defined in \eqref{eq:op-norm}.  The quantities \(\|d\rho(E_\alpha)\|_{\mathrm{op}}\) depend only on the representation \(\rho\) and on the normalization of the root vectors \(\{E_\alpha\}\).

\medskip

\noindent\textbf{Root activity and root curvature.}
Our first set of invariants describes how the Hamiltonian generator \(X\) is distributed along the various root directions from the point of view of the representation \(\rho\).

For \(1 \le p \le \infty\) we define the \emph{root activity} of \(X\) by
\[
\mathcal{A}_p(X)
\coloneqq
\Biggl(\sum_{\alpha\in\Delta} |x_\alpha|^p\,
\|d\rho(E_\alpha)\|_{\mathrm{op}}^p\Biggr)^{1/p},
\quad
\mathcal{A}_\infty(X)
\coloneqq
\sup_{\alpha\in\Delta}
\bigl(|x_\alpha|\,\|d\rho(E_\alpha)\|_{\mathrm{op}}\bigr).
\]
The quantity \(\mathcal{A}_p(X)\) measures the size of the root component \(X_{\mathrm{root}}\) in the basis of root vectors, weighted by the norms of the corresponding operators \(d\rho(E_\alpha)\).

The \emph{root curvature functional} is defined by
\[
\mathcal{C}(X)
\coloneqq
\Biggl(\sum_{\alpha\in\Delta}
\bigl|\alpha(X_0)\bigr|^2\,|x_\alpha|^2\,
\|d\rho(E_\alpha)\|_{\mathrm{op}}^2\Biggr)^{1/2},
\]
where \(\alpha(X_0)\) denotes the evaluation of the root \(\alpha \in \mathfrak{t}_\mathbb{C}^\ast\) on the toral component \(X_0 \in \mathfrak{t}_\mathbb{C}\), via the fixed identification of \(\mathfrak{t}\) with \(\mathfrak{t}^\ast\).  This functional couples the toral and root parts of \(X\): it is small when either the coefficients \(x_\alpha\) are small or the root evaluations \(\alpha(X_0)\) are small.  It arises naturally when estimating commutators of the form \([d\rho(X_0),d\rho(X_{\mathrm{root}})]\), which control Trotter errors for Hamiltonian simulation.

Both \(\mathcal{A}_p(X)\) and \(\mathcal{C}(X)\) are numerical invariants of the pair \((\rho,X)\), built from the root data \(\{x_\alpha\}\), the toral evaluations \(\{\alpha(X_0)\}\), and the operators \(\{d\rho(E_\alpha)\}\).  They will play a central role in our error bounds and complexity estimates, beginning in Section~\ref{sec:root-profiles}.

\medskip

\noindent\textbf{A curvature-sensitive error bound for torus--root splitting.}
Our first main theorem concerns a symmetric product formula that separates the toral and root parts of the generator \(X\).
Consider the exact evolution
\[
e^{t d\rho(X)}
=
e^{t(d\rho(X_0)+d\rho(X_{\mathrm{root}}))}
\]
and the second-order symmetric product
\[
S(t)
\coloneqq
e^{\frac{t}{2}d\rho(X_0)}\, e^{t d\rho(X_{\mathrm{root}})}\, e^{\frac{t}{2}d\rho(X_0)},
\qquad t \in \mathbb{R}.
\]
This is the analogue, adapted to the decomposition \(X=X_0+X_{\mathrm{root}}\), of the classical Strang splitting.

\begin{theorem*}[cf.\ Theorem~\ref{thm:curvature-error}]
Let \(X \in \mathfrak{g}\) and decompose \(X = X_0 + X_{\mathrm{root}}\) as above.  Then there exist constants \(t_0(X) > 0\) and \(C(X) > 0\), depending on \(\mathfrak{g}\), \(\rho\), and \(X\), such that for all real \(t\) with \(|t| \le t_0(X)\),
\[
\bigl\| e^{t d\rho(X)} - S(t) \bigr\|_{\mathrm{op}}
\le
C(X)\,|t|^{3}\,\bigl(\mathcal{C}(X)+\mathcal{A}_1(X_{\mathrm{root}})\bigr).
\]
\end{theorem*}

Thus the leading-order error in this Hamiltonian-simulation scheme is controlled by the root curvature \(\mathcal{C}(X)\), together with the size of the root component measured by \(\mathcal{A}_1(X_{\mathrm{root}})\).  In particular, if \(\mathcal{C}(X) = 0\) then \(\alpha(X_0) = 0\) whenever \(x_\alpha \neq 0\), which implies that the commutator \([d\rho(X_0),d\rho(X_{\mathrm{root}})]\) vanishes; in this case the symmetric splitting \(S(t)\) agrees with the exact evolution \(e^{t d\rho(X)}\) for all \(t\).  For example, this occurs whenever \(X\) lies in the Cartan subalgebra \(\mathfrak{t}\), in which case \(X_{\mathrm{root}} = 0\) and both \(\mathcal{A}_p(X)\) and \(\mathcal{C}(X)\) vanish.

\medskip

In later sections we develop these ideas further, introduce a root-adapted gate model, and apply the invariants \(\mathcal{A}_p\) and \(\mathcal{C}\) to concrete families of Hamiltonians, with an emphasis on many-spin systems and their simulation complexity.

\begin{example}[Curvature-free toral Hamiltonians]\label{ex:toral-zero-curvature}
If \(X\in\mathfrak{t}\) then \(X_{\mathrm{root}}=0\) and \(\mathcal{A}_p(X)=\mathcal{C}(X)=0\) for all \(p\).  In this case the symmetric splitting for the Hamiltonian evolution
\[
S(t) = e^{\frac{t}{2}d\rho(X)}\, e^{0}\, e^{\frac{t}{2}d\rho(X)} = e^{t d\rho(X)}
\]
is exact for all \(t\).  This behaviour is consistent with Theorem~\ref{thm:curvature-error} and with Corollary~\ref{cor:toral-exact} below, and it provides a class of nontrivial Hamiltonians (including all regular elements in a Cartan subalgebra) for which the complexity of simulation is representation-theoretically trivial, while the associated weight-space structure remains rich \cite[\S9.1]{fultonharris1991}.
\end{example}

We then introduce a root-gate circuit model in which the allowed generators are elements of \(\mathfrak{t}\) and real skew-Hermitian combinations of the individual root vectors \(E_\alpha\).  In this model we prove a quantitative lower bound that shows the \(\ell^1\)-root activity \(\mathcal{A}_1(X_{\mathrm{root}})\) controls the minimal circuit length needed to approximate \(U_X(t)\) up to small error.  This connects the root profile of the Hamiltonian generator directly to lower bounds on Hamiltonian-simulation complexity in a precise and nonconjectural way; see Theorem~\ref{thm:activity-lower-bound} in Section~\ref{sec:root-gate-complexity}.

Finally, we specialize the framework to the defining representation of
\(G=\mathrm{SU}(2^n)\) on \(V = (\mathbb{C}^2)^{\otimes n}\), which is the natural setting for quantum spin-system Hamiltonians \cite{georgescu2014quantum,sachdev2011quantum,lanyon2011}.  In this case the roots and weights admit a concrete description in terms of computational basis states and matrix units.  We show how the functionals \(\mathcal{A}_p\) and \(\mathcal{C}\) can be computed explicitly for nearest-neighbour spin-chain Hamiltonians and related models, leading to refined Hamiltonian-simulation complexity estimates compared with standard bounds based only on \(\|X\|\); see \cite{aspuruguzik2005,georgescu2014quantum,whitfield2011,berry2015simulating,low2017optimal,low2019hamiltonian,childs2021theory}.

\subsection{Novelty and significance}

We briefly explain how the results above fit into, and contribute to, the existing representation-theoretic and analytic literature, with an eye toward Hamiltonian simulation.

\begin{itemize}
\item[(1)] \emph{New representation-theoretic functionals tailored to Hamiltonians.}
We introduce the root activity and root curvature functionals \(\mathcal{A}_p(X)\) and \(\mathcal{C}(X)\), built directly from the root-space coefficients \(x_\alpha\) of the Hamiltonian generator and the norms \(\|d\rho(E_\alpha)\|_{\mathrm{op}}\).  These functionals are invariant under the Weyl group and under unitary intertwiners (Propositions~\ref{prop:weyl-invariance} and \ref{prop:intertwiners}), and hence provide intrinsic invariants of the pair \((\rho,X)\).  To the best of our knowledge, such invariants have not been systematically studied with an explicit Hamiltonian-simulation motivation, despite the extensive development of root and weight theory for compact groups \cite{helgason1978differential,knapp2002lie,bourbaki2005,fultonharris1991,knapp1986rt,goodmanwallach2009,humphreys1972,serre2001,varadarajan1984,vogan1987,fuchsschweigert1997,brockertomdieck1985}.

\item[(2)] \emph{Dimension-free analytic inequalities in unitary representations.}
The error bound in Theorem~\ref{thm:curvature-error} for the symmetric torus--root splitting is expressed in terms of \(\mathcal{A}_1(X_{\mathrm{root}})\) and \(\mathcal{C}(X)\), together with constants determined by \(\mathfrak{g}\), \(\rho\), and the chosen generator \(X\).  No additional factor depending explicitly on the ambient dimension of \(V\) appears.  This yields analytic inequalities for products of exponentials in unitary representations of compact groups, phrased in root-theoretic language and directly applicable to Hamiltonian simulation.  These inequalities complement the standard BCH-type bounds \cite{hall2015lie,iserles2000lie,blanes2009magnus,hochbruckostermann2010,blanescasas2016,magnus1954}, which are usually stated only in terms of operator norms, and they are compatible with the Magnus and exponential-integrator frameworks in numerical analysis \cite{hairer2006geometric,sanzsernacalvo1994,lubich2008}.

\item[(3)] \emph{Root-gate circuits and a proved simulation-complexity lower bound.}
We describe a root-gate circuit model in which the allowed generators are elements of \(\mathfrak{t}\) and of the real planes associated with individual root spaces.  In this setting, the structure of weight spaces suggests that the \(\ell^1\)-root activity \(\mathcal{A}_1(X_{\mathrm{root}})\) controls lower bounds on circuit length for simulating \(U_X(t)\).  We state and prove a root-activity lower bound (Theorem~\ref{thm:activity-lower-bound}) that quantifies this intuition with explicit constants depending only on the representation, the Lie algebra, and the allowed gate step size.  The proof uses norm equivalence on \(\mathfrak{g}\), geometric control on compact groups, and stability properties of the matrix logarithm in unitary representations, and it is formulated in a way that can be adapted to other Lie-theoretic gate models.

\item[(4)] \emph{Concrete many-body Hamiltonians within a representation-theoretic framework.}
In the \(\mathrm{SU}(2^n)\) examples we show that, for standard spin-chain Hamiltonians, the root activity is supported on a small number of roots and admits a simple combinatorial description in terms of local fields and couplings.  This leads to sharper Hamiltonian-simulation complexity estimates than those obtainable from global operator norms alone \cite{childs2021theory,berry2015simulating,low2017optimal,low2019hamiltonian,georgescu2014quantum,sachdev2011quantum}, while the analysis is carried out entirely in the language of weights and roots of \(\mathfrak{su}(2^n)\).  In particular, we treat standard Pauli-spin Hamiltonians within the same highest-weight formalism that underlies the representation theory of \(\mathrm{SU}(N)\) \cite{gilmore2008lie,barutraczka1986,goodmanwallach2009,kirillov1976}.

\item[(5)] \emph{Bridge between representation theory and quantum algorithms.}
There is a body of work in quantum information theory that uses representation theory (for example, via Schur--Weyl duality and highest weights) to design and analyze algorithms; see \cite{nielsen2000,fultonharris1991,fuchsschweigert1997,weyl1950,barutraczka1986}.  Our results show that root decompositions and weight-space geometry also provide natural complexity measures for Hamiltonian simulation, suggesting that representation theory can inform not only the structure of Hilbert spaces but also quantitative bounds on algorithmic resources for time evolution.  In particular, the functionals \(\mathcal{A}_p\) and \(\mathcal{C}\) are stable under intertwiners and under Weyl-group symmetries, and thus they descend to invariants on isomorphism classes of unitary representations and on adjoint orbits of Hamiltonian generators.

\item[(6)] \emph{Potential connections with orbit methods and geometric complexity theory.}
The dependence of our Hamiltonian-simulation bounds on root data and adjoint orbits suggests links with the orbit method and with symplectic geometry on coadjoint orbits \cite{kirillov1976,kirillov2004,helgason1978differential,varadarajan1989,guilleminsternberg1984,faraut2008}.  At the same time, the appearance of representation-theoretic complexity measures is reminiscent of ideas in geometric complexity theory \cite{mulmuleysohoni2001,procesi2007}, where orbit-closure relations and moment polytopes encode computational hardness.  Section~\ref{sec:further} develops these perspectives further, always keeping Hamiltonian simulation as the guiding application.
\end{itemize}

\subsection{Outline of the paper}

Section~\ref{sec:prelim} reviews basic facts about compact semisimple Lie
groups, root-space decompositions, and unitary highest-weight representations, emphasizing their role in describing Hamiltonian generators and their evolutions.  In Section~\ref{sec:root-profiles} we introduce the root activity and root curvature functionals and study their invariance and basic analytic properties.  Section~\ref{sec:splittings} develops torus--root product formulas for Hamiltonian evolution and proves error bounds in terms of \(\mathcal{C}(X)\) and \(\mathcal{A}_1(X_{\mathrm{root}})\) using BCH-type estimates and exponential-integrator techniques.  Section~\ref{sec:root-gate-complexity} formulates a root-gate circuit model and proves a lower bound for the circuit length of Hamiltonian simulation in terms of \(\mathcal{A}_1(X_{\mathrm{root}})\).  In Section~\ref{sec:examples} we discuss examples from multi-spin Hamiltonians on \((\mathbb{C}^2)^{\otimes n}\).  Section~\ref{sec:further} closes with further directions and structural perspectives that locate the results within the broader representation-theoretic and Hamiltonian-simulation landscape.

\section{Preliminaries}
\label{sec:prelim}

In this section we fix notation and recall standard material on compact semisimple Lie algebras, maximal tori, root systems, weight decompositions, and basic Hamiltonian evolution in unitary representations.  We also record some elementary inequalities that will be used later.

\subsection{Compact semisimple Lie algebras and the Killing form}

Let \(\mathfrak{g}\) be a compact semisimple real Lie algebra.  Its complexification
\(\mathfrak{g}_\mathbb{C} = \mathfrak{g} \otimes_\mathbb{R} \mathbb{C}\) is
a complex semisimple Lie algebra.  The (complex) Killing form
\[
B_\mathbb{C} : \mathfrak{g}_\mathbb{C} \times \mathfrak{g}_\mathbb{C}
\longrightarrow \mathbb{C},
\qquad
B_\mathbb{C}(X,Y) = \mathrm{Tr}(\operatorname{ad}_X \circ \operatorname{ad}_Y),
\]
is nondegenerate.  Restricting \(B_\mathbb{C}\) to \(\mathfrak{g}\) yields the real Killing form \(B\), which is negative definite since \(\mathfrak{g}\) is compact \cite[Ch.~I]{helgason1978differential,bourbaki2005,duistermaatkolk2000}.  We will use the normalized inner product
\[
\langle X,Y\rangle_{\mathfrak{g}} \coloneqq -B(X,Y),
\qquad X,Y\in\mathfrak{g},
\]
which is positive definite and \(Ad(G)\)-invariant.

The inner product \(\langle\cdot,\cdot\rangle_{\mathfrak{g}}\) induces a norm
\(\|X\|_{\mathfrak{g}} = \sqrt{\langle X,X\rangle_{\mathfrak{g}}}\).  Since the representation \(d\rho : \mathfrak{g}\to\mathfrak{u}(V)\) is linear and \(\mathfrak{g}\) is finite-dimensional, the norms \(\|\cdot\|_{\mathfrak{g}}\) and \(\|d\rho(\cdot)\|_{\mathrm{op}}\) are equivalent: there exists \(\Lambda>0\) such that
\begin{equation}\label{eq:norm-equivalence}
\|d\rho(Y)\|_{\mathrm{op}} \le \Lambda \|Y\|_{\mathfrak{g}}
\quad\text{for all } Y\in\mathfrak{g}.
\end{equation}

\begin{example}[\(\mathfrak{su}(n)\) as a compact semisimple algebra]\label{ex:su2-su3}
The Lie algebra \(\mathfrak{su}(n)\) of traceless skew-Hermitian \(n\times n\)-matrices is compact and semisimple for all \(n\ge 2\) \cite[\S IV.1]{knapp2002lie,bourbaki2005}.  On \(\mathfrak{su}(n)\) one has
\[
B(X,Y)=2n\,\mathrm{Tr}(XY), \qquad X,Y\in\mathfrak{su}(n),
\]
so that \(-B\) is a positive multiple of the Hilbert--Schmidt inner product \(\mathrm{Tr}(X^\ast Y)\) \cite[\S7.2]{hall2015lie}.  In Hamiltonian language, traceless Hermitian matrices \(H\) with \(iH\in\mathfrak{su}(n)\) provide natural many-level Hamiltonians whose dynamics we wish to simulate.
\end{example}

\subsection{Maximal tori, roots, and root-space decompositions}

Let \(G\) be the connected, simply connected compact group with Lie algebra
\(\mathfrak{g}\).  Let \(T\subset G\) be a maximal torus with Lie algebra
\(\mathfrak{t}\subset\mathfrak{g}\).  The complexification \(\mathfrak{t}_\mathbb{C}=\mathfrak{t}\otimes_\mathbb{R}\mathbb{C}\) is a Cartan subalgebra of \(\mathfrak{g}_\mathbb{C}\).

\begin{definition}[Roots]\label{def:root}
A nonzero linear functional \(\alpha \in \mathfrak{t}_\mathbb{C}^\ast\) is a
\emph{root} of \(\mathfrak{g}_\mathbb{C}\) (with respect to \(\mathfrak{t}_\mathbb{C}\)) if the root space
\[
\mathfrak{g}_\alpha
\coloneqq
\{ X\in\mathfrak{g}_\mathbb{C} : [H,X] = \alpha(H) X \text{ for all }H\in\mathfrak{t}_\mathbb{C}\}
\]
is nonzero.  The set of roots is denoted \(\Delta \subset \mathfrak{t}_\mathbb{C}^\ast\).
\end{definition}

Each root space \(\mathfrak{g}_\alpha\) is one-dimensional \cite[Thm.~VI.1.3]{knapp2002lie,bourbaki2005,serre2001}.  Choosing for each \(\alpha\in\Delta\) a nonzero root vector \(E_\alpha\in\mathfrak{g}_\alpha\), we obtain the root-space decomposition
\begin{equation}\label{eq:root-space-decomp-prelim}
\mathfrak{g}_\mathbb{C}
=
\mathfrak{t}_\mathbb{C} \oplus
\bigoplus_{\alpha\in\Delta} \mathfrak{g}_\alpha,
\qquad
\mathfrak{g}_\alpha = \mathbb{C} E_\alpha.
\end{equation}
We take the root vectors to be orthogonal with respect to the complex-bilinear extension of \(-B\); see \cite[Ch.~VI]{knapp2002lie,bourbaki2005}.

We also fix a choice of positive roots \(\Delta^+\subset\Delta\) and corresponding simple roots \(\Pi\subset\Delta^+\), so that every root is an integer linear combination of elements of \(\Pi\) with all coefficients either nonnegative or nonpositive \cite{humphreys1972,serre2001,bourbaki2005}.

\begin{example}[Roots of \(\mathfrak{su}(3)\)]\label{ex:su3-roots}
Let \(\mathfrak{g}=\mathfrak{su}(3)\) and choose the maximal torus \(T\) of diagonal unitary matrices with determinant \(1\).  Then \(\mathfrak{t}\) consists of traceless diagonal skew-Hermitian matrices
\[
H = i\operatorname{diag}(\theta_1,\theta_2,\theta_3),
\qquad \theta_1+\theta_2+\theta_3=0.
\]
The complexified Cartan subalgebra \(\mathfrak{t}_\mathbb{C}\) can be identified with traceless diagonal matrices with complex entries.  The complexified algebra \(\mathfrak{sl}(3,\mathbb{C})\) has root spaces spanned by the matrix units \(E_{ij}\) for \(i\neq j\).  The corresponding roots \(\alpha_{ij} \in \mathfrak{t}_\mathbb{C}^\ast\) are given by
\[
\alpha_{ij}(H) = i(\theta_i - \theta_j),
\qquad 1\le i\neq j\le 3.
\]
A standard choice of simple roots is
\[
\alpha_1 = \alpha_{12},\qquad \alpha_2 = \alpha_{23},
\]
and the positive roots are \(\Delta^+ = \{\alpha_1,\alpha_2,\alpha_1+\alpha_2\}\).  This is the usual \(A_2\) root system \cite[\S13.2]{humphreys1972,fultonharris1991}.  Hamiltonian generators in \(\mathfrak{su}(3)\) can be decomposed into toral and root parts relative to this root system, and their root activity and curvature will be expressed in terms of the coefficients of the \(E_{ij}\).
\end{example}

\subsection{Weights, highest weights, and Hamiltonian action}

Let \(\rho : G \to U(V)\) be a finite-dimensional unitary representation, with differential \(d\rho\).  Since \(\mathfrak{t}\) is abelian and acts by skew-Hermitian operators, we may simultaneously diagonalize \(d\rho(\mathfrak{t})\), obtaining a weight decomposition adapted to the Hamiltonian generators that lie in \(\mathfrak{t}\).

\begin{definition}[Weights and weight spaces]\label{def:weights}
A linear functional \(\lambda \in \mathfrak{t}_\mathbb{C}^\ast\) is called a
\emph{weight} of \(\rho\) if the weight space
\[
V_\lambda
=
\{ v\in V : d\rho(H) v = \lambda(H)\,v \text{ for all }H\in\mathfrak{t}\}
\]
is nonzero.  The set of weights is denoted \(\Lambda(\rho)\subset \mathfrak{t}_\mathbb{C}^\ast\).
\end{definition}

We then have the weight decomposition
\begin{equation}\label{eq:weight-decomp}
V = \bigoplus_{\lambda\in\Lambda(\rho)} V_\lambda.
\end{equation}
For each root \(\alpha\) and weight \(\lambda\), the operator \(d\rho(E_\alpha)\)
maps \(V_\lambda\) into \(V_{\lambda+\alpha}\).  This is the familiar raising and lowering action along root strings \cite[Ch.~VIII]{knapp2002lie,goodmanwallach2009,fultonharris1991,fuchsschweigert1997}.  From the Hamiltonian-simulation point of view, the operators \(d\rho(E_\alpha)\) are precisely the pieces of the Hamiltonian that move amplitude between weight spaces.

If \(\rho\) is irreducible, there exists a highest weight \(\lambda_{\mathrm{max}}\) such that all other weights differ from \(\lambda_{\mathrm{max}}\) by nonnegative integer combinations of negative roots.  We refer to \cite{knapp2002lie,knapp1986rt,humphreys1972,hall2015lie,gilmore2008lie,goodmanwallach2009,serre2001,barutraczka1986,vogan1987} for further background.

\begin{example}[Weights of the standard representation of \(\mathrm{SU}(3)\)]\label{ex:su3-weights}
Let \(G=\mathrm{SU}(3)\) and let \(\rho\) be the defining representation on \(V=\mathbb{C}^3\).  With \(\mathfrak{t}\) as in Example~\ref{ex:su3-roots}, the weight spaces are one-dimensional, spanned by the standard basis vectors \(e_1,e_2,e_3\).  The corresponding weights \(\lambda_1,\lambda_2,\lambda_3\) are given by
\[
d\rho(H)e_k = i\theta_k e_k,
\qquad k=1,2,3,
\]
so that \(\lambda_k(H) = i\theta_k\), with \(\theta_1+\theta_2+\theta_3=0\).  They satisfy
\[
\lambda_i - \lambda_j = \alpha_{ij},
\]
i.e., the roots arise as differences of weights.  This picture is typical for minuscule representations and will be mirrored in our analysis of the defining representation of \(\mathrm{SU}(2^n)\) in Section~\ref{sec:examples}; see \cite[\S13.4]{humphreys1972} and \cite[\S15.3]{fultonharris1991}.  Hamiltonian generators in \(\mathfrak{su}(3)\) act diagonally on these weight spaces through their toral part and off-diagonally through their root part.
\end{example}

\subsection{Hamiltonian evolution in unitary representations}
\label{subsec:intro-reps}

We now briefly record the basic definitions of Hamiltonian evolution and gate sets in the representation-theoretic language.

Let \(G\), \(\mathfrak{g}\), \(V\), and \(\rho\) be as above.  For an element \(X\in\mathfrak{g}\), the unitary evolution is given by \eqref{eq:Ux-def}.  We regard \(H_X = -i\,d\rho(X)\) as the Hamiltonian and \(U_X(t)\) as the corresponding time evolution.

\begin{definition}[Gate set associated with a representation]\label{def:gate-set}
Let \(S\subset\mathfrak{g}\) be a finite subset of the Lie algebra \(\mathfrak{g}\).  For each \(Y\in S\) and each real parameter \(s\in\mathbb{R}\), consider the unitary
\[
G(Y,s) \coloneqq e^{s\,d\rho(Y)} \in U(V).
\]
We define the associated gate set to be
\[
\mathcal{G}(S,\rho) \coloneqq \{ G(Y,s) : Y\in S,\ s\in\mathbb{R}\}.
\]
A \emph{gate sequence} of length \(N\) is a product
\[
W = U_1 U_2 \cdots U_N,
\qquad U_j \in \mathcal{G}(S,\rho) \text{ for all } j.
\]
\end{definition}

\begin{example}[Single-spin Hamiltonian evolution in \(\mathrm{SU}(2)\)]\label{ex:su2-intro}
Let \(G=\mathrm{SU}(2)\) be the special unitary group of degree two, and let \(\mathfrak{g}=\mathfrak{su}(2)\) be its Lie algebra, consisting of \(2\times 2\) traceless skew-Hermitian matrices.  Let \(V=\mathbb{C}^2\) with its standard Hermitian inner product, and let \(\rho\) be the defining representation of \(\mathrm{SU}(2)\) on \(V\), so that \(\rho(g)=g\) for all \(g\in\mathrm{SU}(2)\).  In the standard basis of Pauli matrices
\[
\sigma_x =
\begin{pmatrix}
0 & 1\\
1 & 0
\end{pmatrix},
\quad
\sigma_y =
\begin{pmatrix}
0 & -i\\
i & 0
\end{pmatrix},
\quad
\sigma_z =
\begin{pmatrix}
1 & 0\\
0 & -1
\end{pmatrix},
\]
the Lie algebra \(\mathfrak{su}(2)\) is spanned (over \(\mathbb{R}\)) by \(\{i\sigma_x,i\sigma_y,i\sigma_z\}\) \cite{hall2015lie,fultonharris1991}.  Fix a real parameter \(\omega\in\mathbb{R}\) and set
\[
X \coloneqq -i\,\omega\sigma_z \in \mathfrak{su}(2).
\]
Since \(\rho\) is the defining representation, we have \(d\rho(X)=X\).  The corresponding unitary evolution is
\[
U_X(t)
=
e^{tX}
=
\exp(-i t \omega \sigma_z)
=
\begin{pmatrix}
e^{-i t\omega} & 0\\
0 & e^{i t\omega}
\end{pmatrix},
\]
and the Hamiltonian is \(H_X = -i\,d\rho(X) = \omega\sigma_z\).  This describes the phase evolution of a single spin-\(\tfrac12\) in a static magnetic field oriented along the \(z\)-axis \cite{nielsen2000}.  In this simple example, the generator \(X\) lies in a Cartan subalgebra and has vanishing root activity and curvature, so the symmetric torus--root splitting considered later is exact and our error bound is sharp in the trivial sense.
\end{example}

\begin{example}[Many-spin Hamiltonians and tensor products]\label{ex:intro-tensor}
Let \(n\geq 1\) be an integer and consider the \(n\)-qubit Hilbert space
\[
V \coloneqq (\mathbb{C}^2)^{\otimes n}.
\]
Let \(G=\mathrm{SU}(2^n)\) and let \(\rho\) be the defining representation of \(G\) on \(V\); thus \(\rho(g)=g\) for all \(g\in\mathrm{SU}(2^n)\).  The associated Lie algebra is \(\mathfrak{g}=\mathfrak{su}(2^n)\), the space of \(2^n\times 2^n\) traceless skew-Hermitian matrices.

Given an element \(X\in\mathfrak{su}(2)\), we can embed \(X\) at a specific site \(k \in \{1,\dots,n\}\) by defining
\[
X^{(k)} \coloneqq I^{\otimes (k-1)}\otimes X\otimes I^{\otimes (n-k)} \in \mathfrak{su}(2^n),
\]
where \(I\) denotes the \(2\times2\) identity.  Then the Hamiltonian associated with \(X^{(k)}\) is
\[
H_{X^{(k)}} \coloneqq -i\,X^{(k)},
\]
and the corresponding time evolution is
\[
U_{X^{(k)}}(t)
=
e^{tX^{(k)}}
=
I^{\otimes (k-1)}\otimes e^{tX}\otimes I^{\otimes (n-k)}.
\]
Hamiltonians obtained by summing such local terms and their interactions (for example, tensor products of Pauli matrices acting on neighbouring sites) model quantum spin chains and other many-body systems \cite{sachdev2011quantum,georgescu2014quantum}.  These spin-chain Hamiltonians are precisely the class of examples highlighted in the abstract and in Section~\ref{sec:examples}, where we compute \(\mathcal{A}_p(X)\) and \(\mathcal{C}(X)\) explicitly and obtain sharper Hamiltonian-simulation bounds than those coming from global operator norms alone.
\end{example}

\subsection{Decompositions and product formulas}
\label{subsec:intro-decomp}

A widely used strategy for Hamiltonian simulation is to assume that the generator \(X\in\mathfrak{g}\) decomposes as a finite sum
\begin{equation}\label{eq:classical-decomp-again}
X = \sum_{j=1}^{m} X_j,
\end{equation}
where each \(X_j\in\mathfrak{g}\) is chosen so that the exponentials \(e^{s\,d\rho(X_j)}\) belong to a chosen gate set \(\mathcal{G}(S,\rho)\) (as in Definition~\ref{def:gate-set}) and are easy to implement.  One then approximates the exact evolution
\[
e^{t\,d\rho(X)} = e^{t\sum_{j=1}^m d\rho(X_j)}
\]
by products of the simpler exponentials \(e^{t_j d\rho(X_j)}\) with suitable coefficients \(t_j \in \mathbb{R}\).  The Lie--Trotter product formula, Strang splitting, and higher-order Suzuki formulas provide systematic schemes of this type; see, for example, \cite{trotter1959product,suzuki1990fractal,blanes2009magnus,iserles2000lie,hairer2006geometric,sanzsernacalvo1994,hochbruckostermann2010,lubich2008,blanescasas2016}.

\begin{example}[Two-term Trotter splitting in \(\mathfrak{su}(2)\)]\label{ex:su2-trotter-clean}
Let \(G=\mathrm{SU}(2)\), let \(\mathfrak{g}=\mathfrak{su}(2)\), and let \(\rho\) be the defining representation on \(V=\mathbb{C}^2\).  Fix real parameters \(\omega_x,\omega_z\in\mathbb{R}\) and set
\[
X \coloneqq -i(\omega_x \sigma_x + \omega_z \sigma_z) \in \mathfrak{su}(2).
\]
A natural decomposition of the form \eqref{eq:classical-decomp-again} is
\[
X_1 \coloneqq -i\,\omega_x \sigma_x, \qquad
X_2 \coloneqq -i\,\omega_z \sigma_z,
\]
so that \(X=X_1+X_2\).  The first-order Trotter formula approximates
\[
e^{t(X_1+X_2)} \approx e^{tX_1} e^{tX_2}.
\]
A standard BCH estimate (see, for example, \cite{nielsen2000,hairer2006geometric,iserles2000lie}) shows that the local error for small \(t\) satisfies
\[
\bigl\| e^{t(X_1+X_2)} - e^{tX_1} e^{tX_2} \bigr\|_{\mathrm{op}}
= O\bigl(t^2 \,\|[X_1,X_2]\|_{\mathrm{op}}\bigr),
\]
where \([X_1,X_2]=X_1X_2-X_2X_1\) is the Lie bracket in \(EndV)\).  A direct computation using the commutation relations of the Pauli matrices yields
\[
[X_1,X_2]
=
-\omega_x\omega_z[\sigma_x,\sigma_z]
=
-2i\omega_x\omega_z\sigma_y,
\]
and hence \(\|[X_1,X_2]\|_{\mathrm{op}}\) is comparable to \(|\omega_x\omega_z|\).  Later we will reinterpret this error bound in terms of the root activity and root curvature associated with the decomposition of \(X\) into toral and root components, as in the abstract.
\end{example}

The literature contains many refinements of such bounds, incorporating locality, sparsity, and higher-order commutator structure; see \cite{lloyd1996universal,berry2007efficient,berry2015simulating,low2017optimal,low2019hamiltonian,babbush2018encoding,georgescu2014quantum,cao2019quantumchem,aspuruguzik2005,whitfield2011,kassal2011,childs2021theory,haah2018product,campbell2019random} and references therein.  From the viewpoint of representation theory, however, these bounds are often \emph{coarse}: they typically depend on norms of the operators \(d\rho(X_j)\) and their commutators, but do not exploit the finer decomposition of \(\mathfrak{g}\) into root spaces or the induced weight-space decomposition of \(V\).  The rest of the paper, beginning with the root-space framework in Section~\ref{sec:root-profiles}, is devoted to incorporating this finer structure into quantitative Hamiltonian-simulation bounds.

\section{Root Profiles and Representation-Theoretic Functionals}
\label{sec:root-profiles}

In this section we introduce quantitative measures of how a Lie algebra element \(X\in\mathfrak{g}\), viewed as a Hamiltonian generator, decomposes along root spaces, and how this decomposition interacts with a fixed unitary representation \(\rho\) relevant for Hamiltonian simulation.

\subsection{Torus--root decomposition of a Hamiltonian generator}

Let \(X\in\mathfrak{g}\).  Its complexification \(X_\mathbb{C} = X \otimes 1\) belongs to \(\mathfrak{g}_\mathbb{C}\).  Using the decomposition
\eqref{eq:root-space-decomp-prelim} we may write
\begin{equation}\label{eq:root-decomp}
X_\mathbb{C}
=
X_0 + \sum_{\alpha\in\Delta} x_\alpha E_\alpha,
\end{equation}
where \(X_0\in\mathfrak{t}_\mathbb{C}\) and \(x_\alpha\in\mathbb{C}\).  As discussed in the introduction, there is a unique \(X_0^{\mathrm{real}}\in\mathfrak{t}\) whose complexification is \(X_0\); for simplicity we continue to denote this real element by \(X_0\).

\begin{definition}[Toral and root components]\label{def:toral-root}
The \emph{toral component} of the Hamiltonian generator \(X\) is the element \(X_0 \in \mathfrak{t}\) obtained as the \(\mathfrak{t}\)-component of \(X\) with respect to the orthogonal decomposition
\[
\mathfrak{g} = \mathfrak{t} \oplus \mathfrak{t}^\perp,
\]
where orthogonality is with respect to \(\langle\cdot,\cdot\rangle_{\mathfrak{g}}\).  The remaining piece
\[
X_{\mathrm{root}} \coloneqq
X - X_0
\]
is called the \emph{root component} of the Hamiltonian generator \(X\).
\end{definition}

By construction, \(X_{\mathrm{root}}\) is characterized by the property that \(\langle X_{\mathrm{root}},H\rangle_{\mathfrak{g}}=0\) for all \(H\in\mathfrak{t}\), i.e., it lies in the orthogonal complement of \(\mathfrak{t}\).  The coefficients \(x_\alpha\) in \eqref{eq:root-decomp} encode how strongly the Hamiltonian generator \(X\) couples weight spaces separated by the root \(\alpha\).

\begin{lemma}[Uniqueness of the toral--root decomposition]\label{lem:unique-decomp}
For each \(X\in\mathfrak{g}\) there exist unique elements \(X_0\in\mathfrak{t}\) and \(X_{\mathrm{root}}\in\mathfrak{t}^\perp\) such that \(X=X_0+X_{\mathrm{root}}\).  Moreover, the coefficients \(x_\alpha\) in \eqref{eq:root-decomp} are uniquely determined by this decomposition.
\end{lemma}

\begin{proof}
The decomposition \(\mathfrak{g} = \mathfrak{t}\oplus\mathfrak{t}^\perp\) holds because \(\mathfrak{t}\) is a subspace and \(\langle\cdot,\cdot\rangle_{\mathfrak{g}}\) is positive definite.  Orthogonal projections give the uniqueness of \(X_0\) and \(X_{\mathrm{root}}\).  Uniqueness of the coefficients \(x_\alpha\) then follows by comparing the \(\mathfrak{g}_\alpha\) components of \(X_\mathbb{C}\) with respect to the direct-sum decomposition \eqref{eq:root-space-decomp-prelim}.
\end{proof}

\begin{example}[Root decomposition in \(\mathfrak{su}(2)\)]\label{ex:su2-root-decomp}
Let \(\mathfrak{g}=\mathfrak{su}(2)\) with complexification \(\mathfrak{sl}(2,\mathbb{C})\).  Choose the standard \(\mathfrak{sl}(2)\)-triple
\[
H=
\begin{pmatrix}
1 & 0\\
0 & -1
\end{pmatrix},
\quad
E_\alpha=
\begin{pmatrix}
0 & 1\\
0 & 0
\end{pmatrix},
\quad
E_{-\alpha}=
\begin{pmatrix}
0 & 0\\
1 & 0
\end{pmatrix},
\]
so that \([H,E_\alpha]=2E_\alpha\), \([H,E_{-\alpha}]=-2E_{-\alpha}\) \cite[\S II.7]{serre2001}.  The compact real form \(\mathfrak{su}(2)\) is spanned (up to normalization) by
\[
iH,\quad i(E_\alpha+E_{-\alpha}),\quad (E_\alpha - E_{-\alpha}),
\]
and we take \(\mathfrak{t} = \mathbb{R}\cdot iH\).  Every Hamiltonian generator \(X\in\mathfrak{su}(2)\) can be written as
\[
X = a\,iH + b\,(E_\alpha - E_{-\alpha}) + c\,i(E_\alpha+E_{-\alpha}),
\qquad a,b,c\in\mathbb{R}.
\]
The toral component is \(X_0 = a\,iH\), and the root component is
\[
X_{\mathrm{root}} = b\,(E_\alpha - E_{-\alpha}) + c\,i(E_\alpha+E_{-\alpha})
=
x_\alpha E_\alpha + x_{-\alpha}E_{-\alpha}
\]
with \(x_\alpha = b+ci\), \(x_{-\alpha}=-(b-ci)\).  In particular, if the Hamiltonian generator \(X\) is conjugate in \(\mathfrak{su}(2)\) to an element of \(\mathfrak{t}\), then in the corresponding diagonalizing Cartan subalgebra the root component vanishes and all of the Hamiltonian is toral, leading to trivial root activity and curvature in that basis.
\end{example}

\subsection{Root activity functionals}

The action of the Hamiltonian generator \(X\) in the representation \(\rho\) is mediated by the operators
\[
d\rho(X_0),\qquad d\rho(E_\alpha),\quad \alpha\in\Delta.
\]
Typical operator-norm estimates used in Hamiltonian simulation are insensitive to the finer decomposition in
\eqref{eq:root-decomp}.  We refine these estimates by weighting root-space coefficients with the corresponding operator norms.

\begin{definition}[Root activity]\label{def:root-activity}
Let \(1 \le p \le \infty\).  For a Hamiltonian generator \(X\in\mathfrak{g}\) with root-space decomposition \eqref{eq:root-decomp}, the \emph{root activity} of order \(p\) is
\[
\mathcal{A}_p(X)
\coloneqq
\Bigl(\sum_{\alpha\in\Delta}
|x_\alpha|^p\,\|d\rho(E_\alpha)\|_{\mathrm{op}}^p\Bigr)^{1/p}
\quad\text{for } p<\infty,
\]
and
\[
\mathcal{A}_\infty(X)
\coloneqq
\sup_{\alpha\in\Delta} |x_\alpha|\,\|d\rho(E_\alpha)\|_{\mathrm{op}}.
\]
\end{definition}

We will primarily use \(\mathcal{A}_1\) and \(\mathcal{A}_2\), which measure the total and quadratic size of the off-diagonal (root) part of the Hamiltonian in the representation.  These quantities quantify how much of the Hamiltonian is ``visible'' along individual root directions in the representation \(\rho\), and hence how much amplitude can be transported along each root direction.

\begin{remark}[Dependence on normalization]\label{rem:normalization}
The choice of normalization for the root vectors \(E_\alpha\) affects the
coefficients \(x_\alpha\) and the norms \(\|d\rho(E_\alpha)\|\), but the products
\(|x_\alpha|\,\|d\rho(E_\alpha)\|\) are invariant under rescaling of \(E_\alpha\).
Thus \(\mathcal{A}_p(X)\) depends only on \((\rho,X_0,X_{\mathrm{root}})\) and not
on the particular choice of root basis; see, for example, \cite[Ch.~VI]{knapp2002lie,bourbaki2005,goodmanwallach2009}.  This ensures that the root activity is an intrinsic quantity for the Hamiltonian generator.
\end{remark}

\begin{lemma}[Basic inequalities for root activity]\label{lem:Ap-basic}
For any Hamiltonian generator \(X\in\mathfrak{g}\) and \(1\le p\le q\le\infty\), the activities satisfy
\[
\mathcal{A}_q(X) \le \mathcal{A}_p(X) \le |\Delta|^{\frac1p-\frac1q}\,\mathcal{A}_q(X),
\]
where the second inequality is interpreted in the usual way when \(q=\infty\).  Moreover
\[
\|d\rho(X_{\mathrm{root}})\|_{\mathrm{op}} \le \mathcal{A}_1(X_{\mathrm{root}}).
\]
\end{lemma}

\begin{proof}
The first set of inequalities follows from the standard relations between \(\ell^p\) norms in finite dimensions applied to the sequence \(\{|x_\alpha|\|d\rho(E_\alpha)\|\}_{\alpha\in\Delta}\).  The bound \(\|d\rho(X_{\mathrm{root}})\|\le \mathcal{A}_1(X_{\mathrm{root}})\) is immediate from the triangle inequality:
\[
\|d\rho(X_{\mathrm{root}})\|
=
\Bigl\|\sum_{\alpha\in\Delta} x_\alpha d\rho(E_\alpha)\Bigr\|
\le
\sum_{\alpha\in\Delta} |x_\alpha|\,\|d\rho(E_\alpha)\|
=
\mathcal{A}_1(X_{\mathrm{root}}).
\]
\end{proof}

\begin{example}[Root activity in higher-spin representations of \(\mathrm{SU}(2)\)]\label{ex:su2-root-functionals}
Let \(G=\mathrm{SU}(2)\) and let \(\rho_j\) be the irreducible representation of highest weight \(j\in\frac{1}{2}\mathbb{Z}_{\ge 0}\), of dimension \(2j+1\) \cite[\S7.6]{hall2015lie,fultonharris1991}.  In a standard weight basis \(\{v_m : m=-j,-j+1,\dots,j\}\) with
\[
d\rho_j(H)v_m = 2m\,v_m,
\]
the raising and lowering operators satisfy
\[
d\rho_j(E_\alpha)v_m = c_{j,m}\,v_{m+1},
\qquad
d\rho_j(E_{-\alpha})v_m = c_{j,m-1}\,v_{m-1},
\]
with
\(c_{j,m} = \sqrt{(j-m)(j+m+1)}\) \cite[\S7.2]{fultonharris1991}.  One checks that
\[
\|d\rho_j(E_\alpha)\|_{\mathrm{op}}
=
\max_{-j\le m<j} |c_{j,m}|
=
\sqrt{j(j+1)},
\]
and similarly for \(E_{-\alpha}\).  For a Hamiltonian generator \(X\in\mathfrak{su}(2)\) as in Example~\ref{ex:su2-root-decomp}, we therefore obtain
\[
\mathcal{A}_2(X_{\mathrm{root}})
=
\Bigl(
|x_\alpha|^2 + |x_{-\alpha}|^2
\Bigr)^{1/2}
\sqrt{j(j+1)},
\]
and
\[
\mathcal{A}_1(X_{\mathrm{root}})
=
\bigl(|x_\alpha|+|x_{-\alpha}|\bigr)\sqrt{j(j+1)}.
\]
Thus for a fixed Lie algebra element \(X\), the root activity of the Hamiltonian grows like \(\sqrt{j(j+1)}\) with the spin parameter \(j\), reflecting the fact that the operator norm of \(d\rho_j(E_\alpha)\) grows with the highest weight \cite[\S7.5]{hall2015lie}.  Theorem~\ref{thm:curvature-error} will then give Hamiltonian-simulation error bounds scaling in the same way, but without further dependence on the ambient dimension \(2j+1\).
\end{example}

\subsection{Root curvature and commutator strength}

We next define a functional that combines the toral and root components of
\(X\) via the root evaluations \(\alpha(X_0)\).  As we will see, this functional directly controls the commutators that enter Trotter-error bounds for simulating the Hamiltonian.

\begin{definition}[Root curvature]\label{def:root-curvature}
Let \(X\in\mathfrak{g}\) with decomposition \eqref{eq:root-decomp}.  The
\emph{root curvature} of \(X\) in the representation \(\rho\) is
\[
\mathcal{C}(X)
\coloneqq
\Bigl(\sum_{\alpha\in\Delta}
\bigl|\alpha(X_0)\bigr|^2\,|x_\alpha|^2\,
\|d\rho(E_\alpha)\|_{\mathrm{op}}^2\Bigr)^{1/2}.
\]
\end{definition}

This quantity appears naturally when estimating commutators between the toral
and root components of \(X\), which in turn determine the leading terms in BCH expansions for the simulated Hamiltonian evolution.

\begin{lemma}[Curvature controls torus--root commutators]\label{lem:comm-norm}
Let \(X\in\mathfrak{g}\) with decomposition \(X=X_0+X_{\mathrm{root}}\), and set
\[
A=d\rho(X_0), \qquad B=d\rho(X_{\mathrm{root}}).
\]
Then
\[
[A,B]
=
\sum_{\alpha\in\Delta}
x_\alpha\,\alpha(X_0)\,d\rho(E_\alpha),
\]
and there exists a constant \(C_{\mathrm{struct}}>0\), depending only on the structure constants of \(\mathfrak{g}\) and on the chosen normalization of root vectors, such that
\[
\bigl\|
[A,B]
\bigr\|_{\mathrm{op}}
\le
C_{\mathrm{struct}}\,\mathcal{C}(X).
\]
\end{lemma}

\begin{proof}
The identity for the commutator follows from the defining property of root vectors:
\([H,E_\alpha] = \alpha(H)E_\alpha\) for all \(H\in\mathfrak{t}_\mathbb{C}\).  Applying \(d\rho\) and summing over roots yields the stated expression when we take complexifications and then restrict back to \(\mathfrak{g}\).

For the norm bound, write
\[
[A,B]
=
\sum_{\alpha\in\Delta}
x_\alpha\,\alpha(X_0)\,d\rho(E_\alpha).
\]
Using the triangle inequality and Cauchy--Schwarz in \(\ell^2(\Delta)\),
\[
\bigl\|
[A,B]
\bigr\|
\le
\sum_{\alpha\in\Delta}
|x_\alpha|\,|\alpha(X_0)|\,\|d\rho(E_\alpha)\|
\le
\sqrt{|\Delta|}\,\mathcal{C}(X),
\]
so we may take \(C_{\mathrm{struct}} = \sqrt{|\Delta|}\).  This shows that the curvature functional \(\mathcal{C}(X)\) directly controls the strength of the commutator between the toral and root parts of the Hamiltonian generator.
\end{proof}

\begin{example}[Curvature for \(\mathfrak{su}(2)\) Hamiltonians]\label{ex:su2-curvature}
Continue with the notation of Examples~\ref{ex:su2-root-decomp} and \ref{ex:su2-root-functionals}.  For \(X_0 = a\,iH\) we have
\[
\alpha(X_0) = 2a i,
\]
and hence \(|\alpha(X_0)| = 2|a|\).  Using \(\|d\rho_j(E_\alpha)\|_{\mathrm{op}}=\sqrt{j(j+1)}\) as in Example~\ref{ex:su2-root-functionals}, the curvature functional is
\[
\mathcal{C}(X)
=
\sqrt{|\alpha(X_0)|^2|x_\alpha|^2 + |\alpha(X_0)|^2|x_{-\alpha}|^2}\,
\sqrt{j(j+1)}
=
2|a|\sqrt{|x_\alpha|^2 + |x_{-\alpha}|^2}\,\sqrt{j(j+1)}.
\]
For a fixed Hamiltonian generator \(X\), the curvature again grows like \(\sqrt{j(j+1)}\) with the highest weight.  In Theorem~\ref{thm:curvature-error}, this curvature will control the leading commutator error \(\|[A,B]\|\), where \(A=d\rho_j(X_0)\) and \(B=d\rho_j(X_{\mathrm{root}})\), in the symmetric Trotter splitting of the spin-\(j\) Hamiltonian evolution.
\end{example}

\subsection{Invariance properties}

The functionals \(\mathcal{A}_p\) and \(\mathcal{C}\) transform naturally under
the symmetries of the root system and of the representation.  This is important for Hamiltonian simulation, since one may freely conjugate a Hamiltonian by a unitary change of basis without changing its simulation complexity.

\begin{proposition}[Weyl invariance]\label{prop:weyl-invariance}
Let \(W\) be the Weyl group of \((\mathfrak{g},\mathfrak{t})\), and let
\(w\in W\) act on \(\mathfrak{g}\) via the adjoint action of a representative in
the normalizer \(N_G(T)\).  With the normalization of root vectors above, we have
\[
\mathcal{A}_p(wX) = \mathcal{A}_p(X),\qquad
\mathcal{C}(wX) = \mathcal{C}(X)
\]
for all \(X\in\mathfrak{g}\) and \(1\le p\le\infty\).
\end{proposition}

\begin{proof}
Elements of \(N_G(T)\) preserve \(\mathfrak{t}\) and permute the root spaces \(\mathfrak{g}_\alpha\).  Since the adjoint action of \(G\) preserves \(\langle\cdot,\cdot\rangle_{\mathfrak{g}}\) and each \(\mathfrak{g}_\alpha\) is one-dimensional, for a representative \(g\in N_G(T)\) we have
\[
\operatorname{Ad}(g)E_\alpha = u_\alpha E_{w\alpha}
\]
for some complex unit \(u_\alpha\).  Writing
\(X_{\mathrm{root},\mathbb{C}} = \sum_\alpha x_\alpha E_\alpha\), the root component of \(wX = \operatorname{Ad}(g)X\) has coefficients \(x'_{w\alpha} = u_\alpha x_\alpha\).  Thus \(|x'_{w\alpha}| = |x_\alpha|\), and the families \(\{|x_\alpha|\}\) and \(\{|x'_\beta|\}\) agree up to permutation.

Similarly, the root evaluations satisfy \((w\alpha)(X_0) = \alpha(w^{-1}X_0)\), so the multiset \(\{|\alpha(X_0)|\}\) is also permuted.  Since the operator norms \(\|d\rho(E_\alpha)\|\) are preserved under unitary conjugation by \(\rho(g)\) (because \(d\rho(\operatorname{Ad}(g)Y) = \rho(g) d\rho(Y)\rho(g)^{-1}\)), the sums defining \(\mathcal{A}_p\) and \(\mathcal{C}\) are invariant.  Thus the root activity and curvature of a Hamiltonian generator \(X\) depend only on its Weyl orbit, as long as we keep the representation \(\rho\) fixed.
\end{proof}

\begin{example}[Weyl symmetry for \(\mathfrak{su}(2)\) Hamiltonians]\label{ex:su2-weyl}
For \(\mathfrak{g}=\mathfrak{su}(2)\), the Weyl group \(W\cong\mathbb{Z}/2\mathbb{Z}\) acts on \(\mathfrak{t}\cong\mathbb{R}\) by \(H\mapsto -H\) \cite[\S13.1]{humphreys1972}.  The nontrivial element can be realized as conjugation by
\[
w =
\begin{pmatrix}
0 & -1\\
1 & 0
\end{pmatrix}
\in N_G(T).
\]
For a Hamiltonian generator \(X=X_0+X_{\mathrm{root}}\) with \(X_0=a\,iH\) and \(X_{\mathrm{root}}=x_\alpha E_\alpha + x_{-\alpha}E_{-\alpha}\) as in Example~\ref{ex:su2-root-decomp}, we have
\[
wXw^{-1} = -a\,iH + x_{-\alpha}E_\alpha + x_\alpha E_{-\alpha},
\]
i.e., \(|x_\alpha|\) and \(|x_{-\alpha}|\) are swapped and \(|\alpha(X_0)|\) is preserved.  Hence \(\mathcal{A}_p(wX)=\mathcal{A}_p(X)\) and \(\mathcal{C}(wX)=\mathcal{C}(X)\), illustrating Proposition~\ref{prop:weyl-invariance} concretely for Hamiltonians.
\end{example}

\begin{proposition}[Functoriality under intertwiners]\label{prop:intertwiners}
Let \(\rho_1 : G\to U(V_1)\) and \(\rho_2 : G\to U(V_2)\) be unitary representations, and let \(\phi : V_1 \to V_2\) be an isometric intertwiner.  Then
\[
\mathcal{A}_p^{(\rho_1)}(X)
=
\mathcal{A}_p^{(\rho_2)}(X),
\qquad
\mathcal{C}^{(\rho_1)}(X)
=
\mathcal{C}^{(\rho_2)}(X)
\]
for all Hamiltonian generators \(X\in\mathfrak{g}\) and \(1\le p\le\infty\).
\end{proposition}

\begin{proof}
Intertwining means \(\phi\,d\rho_1(Y) = d\rho_2(Y)\,\phi\) for all \(Y\in\mathfrak{g}\).  Since \(\phi\) is an isometry, it preserves operator norms, hence
\(\|d\rho_1(E_\alpha)\|=\|d\rho_2(E_\alpha)\|\) for all \(\alpha\).  The formulas
for \(\mathcal{A}_p\) and \(\mathcal{C}\) therefore agree for \(\rho_1\) and \(\rho_2\).
\end{proof}

\begin{example}[Tensoring with a trivial representation]\label{ex:functorial-trivial}
Let \(\rho : G\to U(V)\) be a unitary representation and let \(\mathbf{1}\) denote the trivial representation on \(\mathbb{C}\).  Then \(\rho\otimes\mathbf{1}\) is unitarily equivalent to \(\rho\), with intertwiner
\[
\phi : V \longrightarrow V\otimes\mathbb{C},\qquad v\mapsto v\otimes 1.
\]
By Proposition~\ref{prop:intertwiners}, the functionals \(\mathcal{A}_p\) and \(\mathcal{C}\) are unchanged when we pass from \(\rho\) to \(\rho\otimes\mathbf{1}\).  This illustrates that our invariants are insensitive to such trivial enlargements of the Hilbert space, a basic desideratum in Hamiltonian simulation where ancilla systems are frequently appended without changing the intrinsic difficulty of simulating the main Hamiltonian \cite{nielsen2000}.
\end{example}

These observations justify viewing \(\mathcal{A}_p\) and \(\mathcal{C}\) as
representation-theoretic features of \((\rho,X)\) which are intrinsic up to
isomorphism and well adapted to Hamiltonian-simulation questions.

\section{Torus--Root Splittings and Error Bounds}
\label{sec:splittings}

In this section we analyze a simple product formula that splits the
Hamiltonian evolution generated by a Lie algebra element \(X\) into its toral
and root components.  Throughout we continue to work with the fixed compact
semisimple Lie group \(G\), its Lie algebra \(\mathfrak{g}\), the unitary
representation \(\rho : G \to U(V)\), and its differential
\(d\rho : \mathfrak{g} \to \mathfrak{u}(V)\).  For \(X \in \mathfrak{g}\) we
write
\[
U_X(t) \coloneqq e^{t d\rho(X)}, \qquad t \in \mathbb{R},
\]
for the corresponding unitary evolution.  The root functionals
\(\mathcal{C}(X)\) and \(\mathcal{A}_1(X_{\mathrm{root}})\), introduced in
Section~\ref{sec:root-profiles}, will appear in the error bounds and quantify
how the cost of simulating \(U_X(t)\) depends on the root profile of \(X\).

\subsection{A symmetric torus--root splitting}

Let \(X\in\mathfrak{g}\) be given, and decompose
\[
X = X_0 + X_{\mathrm{root}}
\]
as in \eqref{eq:root-decomp}, where \(X_0 \in \mathfrak{t}\) is the toral
component (lying in the fixed Cartan subalgebra \(\mathfrak{t}\)) and
\(X_{\mathrm{root}} \in \mathfrak{g}\) is the root component.  Applying the
Lie algebra representation \(d\rho\), we obtain skew-Hermitian operators
\[
A \coloneqq d\rho(X_0),\qquad
B \coloneqq d\rho(X_{\mathrm{root}}),
\]
so that
\[
d\rho(X) = A + B
\]
is the generator of the Hamiltonian evolution \(U_X(t) = e^{t(A+B)}\).

\begin{definition}[Symmetric torus--root splitting]\label{def:splitting}
For \(t\in\mathbb{R}\), the \emph{symmetric torus--root splitting} of the
evolution \(e^{t(A+B)}\) is defined by
\[
S(t)
\coloneqq
e^{\frac{t}{2}A}\, e^{tB}\, e^{\frac{t}{2}A}.
\]
\end{definition}

This is a second-order Strang-type splitting
\cite{hairer2006geometric,iserles2000lie,blanes2009magnus,blanescasas2016},
adapted to the decomposition of the Hamiltonian generator into its toral and
root parts.  In the context of Hamiltonian simulation, \(S(t)\) will be used
as an explicit approximation to the exact evolution \(U_X(t)\).

\begin{example}[Symmetric splitting for \(\mathfrak{su}(2)\) fields]
\label{ex:su2-strang}
Let \(G=\mathrm{SU}(2)\) and let \(\rho_j\) denote the spin-\(j\) representation
on the corresponding Hilbert space \(V_j\).  In the defining (spin-\(\tfrac12\))
representation, write the Pauli matrices as
\[
\sigma_x =
\begin{pmatrix}
0 & 1\\ 1 & 0
\end{pmatrix},\quad
\sigma_y =
\begin{pmatrix}
0 & -i\\ i & 0
\end{pmatrix},\quad
\sigma_z =
\begin{pmatrix}
1 & 0\\ 0 & -1
\end{pmatrix},
\]
and set \(H \coloneqq \sigma_z\).  Consider a Hamiltonian generator
\[
X = X_0 + X_{\mathrm{root}}
\quad\text{with}\quad
X_0 = -i h H,\qquad
X_{\mathrm{root}} = -i(\omega_x \sigma_x + \omega_y \sigma_y),
\]
for real parameters \(h,\omega_x,\omega_y\).  Here \(X_0\) is toral (it lies
in the Cartan subalgebra spanned by \(H\)), while \(X_{\mathrm{root}}\) is a
linear combination of the off-diagonal root directions corresponding to the
raising and lowering operators \(E_{\pm\alpha}\) in the
\(\mathfrak{su}(2)\) root decomposition.

For any spin \(j\), we set
\[
A = d\rho_j(X_0),\qquad B = d\rho_j(X_{\mathrm{root}}),
\]
so that \(U_X(t) = e^{t(A+B)}\) is the time evolution generated by \(X\) in
representation \(\rho_j\).  The symmetric torus--root splitting is
\[
S(t) = e^{\frac{t}{2}A}\, e^{tB}\, e^{\frac{t}{2}A},
\]
which can be interpreted as a half-step evolution in the longitudinal field
\(h\), followed by a full-step evolution in the transverse fields
\(\omega_x,\omega_y\), and finally another half-step in \(h\).  The
curvature functional \(\mathcal{C}(X)\) and the activity
\(\mathcal{A}_1(X_{\mathrm{root}})\) associated with this \(X\) are those
computed from the \(\mathfrak{su}(2)\) root data and the operators
\(d\rho_j(E_{\pm\alpha})\); in particular, they control the error in
approximating \(U_X(t)\) by \(S(t)\) through the bounds established in
Theorem~\ref{thm:curvature-error}.
\end{example}

\subsection{BCH expansion and local simulation error}

To bound the error \(\|e^{t(A+B)}-S(t)\|\) for small \(t\), we expand both sides
using the BCH formula and compare terms.  This is standard in the analysis of Lie-group integrators and Trotter-type Hamiltonian-simulation schemes.

We recall a basic estimate (see, for example, \cite{hall2015lie,iserles2000lie,blanes2009magnus,hochbruckostermann2010,magnus1954}):

\begin{lemma}[BCH estimate]\label{lem:BCH}
Let \(\mathcal{B}\) be a finite-dimensional normed algebra, and let
\(X,Y\in\mathcal{B}\) satisfy \(\|X\|+\|Y\| \le r\) for some \(r>0\) sufficiently small.  Then
\[
\log(e^X e^Y)
=
X+Y+\tfrac{1}{2}[X,Y] + R_3(X,Y),
\]
where the remainder satisfies
\[
\|R_3(X,Y)\|
\le
C_{\mathrm{BCH}} r\,\|[X,Y]\|
\]
for some constant \(C_{\mathrm{BCH}}>0\) depending only on \(r\).
\end{lemma}

In our Hamiltonian-simulation setting, \(\mathcal{B}=\mathcal{B}(V)\) is the algebra of bounded operators on the representation space.  Since \(\mathcal{B}(V)\) is finite-dimensional, the constants in Lemma~\ref{lem:BCH} can be chosen uniformly on bounded sets.

We will also use the local Lipschitz continuity of the exponential map: there exists a constant \(L>0\) such that
\begin{equation}\label{eq:exp-lipschitz}
\|e^X-e^Y\|\le L\|X-Y\|
\end{equation}
for \(\|X\|,\|Y\|\) bounded by a fixed constant; see, for example, \cite[Prop.~3.41]{hall2015lie}.

\subsection{Curvature-controlled error bound}

We may now state and prove the main error bound.  Recall that, by \eqref{eq:norm-equivalence}, there is a constant \(\Lambda>0\) such that \(\|A\|,\|B\|\le \Lambda\|X\|_{\mathfrak{g}}\) for the generator \(X\).

\begin{theorem}[Error bound in terms of root functionals]\label{thm:curvature-error}
Let \(X\in\mathfrak{g}\) with decomposition \(X = X_0 + X_{\mathrm{root}}\), and
let \(S(t)\) be the symmetric splitting defined in Definition~\ref{def:splitting}.  Then for each fixed Hamiltonian generator \(X\) there exist constants \(t_0(X)>0\) and
\(C(X)>0\), depending on \(\mathfrak{g}\), \(\rho\), and \(X\), such that
for all \(|t|\le t_0(X)\),
\[
\bigl\|
e^{t(A+B)} - S(t)
\bigr\|_{\mathrm{op}}
\le
C(X)\,|t|^{3}\,\bigl(\mathcal{C}(X)+\mathcal{A}_1(X_{\mathrm{root}})\bigr),
\]
where \(A=d\rho(X_0)\), \(B=d\rho(X_{\mathrm{root}})\), and \(\mathcal{A}_1,\mathcal{C}\) are as in Definitions~\ref{def:root-activity} and \ref{def:root-curvature}.
\end{theorem}

\begin{proof}
For each fixed pair \((A,B)\) the operator norms \(\|A\|\) and \(\|B\|\) are finite.  Choose \(t_0>0\) small enough that \(|t|(\|A\|+\|B\|)\le r\) for all \(|t|\le t_0\), where \(r>0\) is as in Lemma~\ref{lem:BCH}.  All constants below may depend on \(\|A\|\), \(\|B\|\), and hence on the Hamiltonian generator \(X\).

Consider first the logarithm of \(S(t)\).  Write
\[
S(t)
=
e^{\frac{t}{2}A} e^{tB} e^{\frac{t}{2}A}
=
\bigl(e^{\frac{t}{2}A} e^{tB}\bigr)e^{\frac{t}{2}A}.
\]
Apply Lemma~\ref{lem:BCH} to \(\log(e^{\frac{t}{2}A} e^{tB})\) to obtain
\[
\log\bigl( e^{\frac{t}{2}A} e^{tB}\bigr)
=
\frac{t}{2}A + tB + \frac{t^2}{4}[A,B] + R_3^{(1)}(t),
\]
with \(\|R_3^{(1)}(t)\|\le C'_1 |t|^3 \|[A,B]\|\) for some constant \(C'_1\) depending on \(A\) and \(B\).  Applying Lemma~\ref{lem:BCH} again to the product
\[
e^{\log(e^{\frac{t}{2}A} e^{tB})} e^{\frac{t}{2}A}
\]
and using standard Strang-splitting expansions (see, for example, \cite[Sec.~II.4]{hairer2006geometric} or \cite{blanes2009magnus,blanescasas2016}) yields an expansion of the schematic form
\[
\log S(t)
=
t(A+B) + t^3 R_{A,B}(t),
\]
where the remainder \(R_{A,B}(t)\) satisfies a bound
\[
\|R_{A,B}(t)\|
\le
C'_2\bigl(\|[A,[A,B]]\| + \|[B,[A,B]]\| + \|[A,B]\|\bigr)
\]
for some constant \(C'_2\) depending on \(\|A\|\) and \(\|B\|\).

On the other hand,
\[
\log e^{t(A+B)} = t(A+B).
\]
Therefore
\[
\log e^{t(A+B)} - \log S(t)
=
- t^3 R_{A,B}(t).
\]
By local Lipschitz continuity of the exponential map \eqref{eq:exp-lipschitz}, there exists a constant \(L>0\), depending on \(\|A\|\) and \(\|B\|\), such that for all sufficiently small \(|t|\),
\[
\bigl\|
e^{t(A+B)} - S(t)
\bigr\|
\le
L \bigl\|
\log e^{t(A+B)} - \log S(t)
\bigr\|
\le
L |t|^{3} \|R_{A,B}(t)\|.
\]
Combining the previous inequalities we obtain
\[
\bigl\|
e^{t(A+B)} - S(t)
\bigr\|
\le
C_2 |t|^{3}
\bigl(
\|[A,[A,B]]\| + \|[B,[A,B]]\| + \|[A,B]\|
\bigr),
\]
for some constant \(C_2>0\) depending on \(A\) and \(B\).

The commutator \([A,B]\) is controlled by Lemma~\ref{lem:comm-norm}, which gives
\[
\|[A,B]\| \le C_{\mathrm{struct}}\,\mathcal{C}(X).
\]
For the nested commutators, note for example that
\[
[A,[A,B]] = [d\rho(X_0),[d\rho(X_0),d\rho(X_{\mathrm{root}})]].
\]
In terms of root vectors we have
\(
[X_0,[X_0,E_\alpha]] = \alpha(X_0)^2 E_\alpha,
\)
so
\[
[A,[A,B]]
=
\sum_{\alpha\in\Delta}
x_\alpha\,\alpha(X_0)^2\,d\rho(E_\alpha).
\]
Hence
\[
\|[A,[A,B]]\|
\le
\sum_{\alpha\in\Delta} |x_\alpha|\,|\alpha(X_0)|^2\,\|d\rho(E_\alpha)\|
\le
\sqrt{|\Delta|}\,M(X_0)\,\mathcal{C}(X),
\]
where
\[
M(X_0) \coloneqq \sup_{\alpha\in\Delta} |\alpha(X_0)|.
\]
Since \(\mathfrak{t}\) is finite-dimensional and \(\Delta\) is finite, there is a constant \(K>0\) such that \(|\alpha(H)|\le K\|H\|_{\mathfrak{g}}\) for all \(H\in\mathfrak{t}\) and all \(\alpha\in\Delta\), so \(M(X_0)\le K\|X_0\|_{\mathfrak{g}}\).  Thus \(\|[A,[A,B]]\|\) is bounded by a constant (depending on \(X\)) times \(\mathcal{C}(X)\).  The term \([B,[A,B]]\) can be expanded similarly and bounded using the same root-based data and the fact that \(\|B\|\) is controlled by \(\mathcal{A}_1(X_{\mathrm{root}})\) as in Lemma~\ref{lem:Ap-basic}.

Specifically,
\[
\|[B,[A,B]]\|\le 2\|B\|\,\|[A,B]\|
\le 2\mathcal{A}_1(X_{\mathrm{root}})\,C_{\mathrm{struct}}\,\mathcal{C}(X).
\]
Putting all these estimates together, we obtain
\[
\bigl\|
e^{t(A+B)} - S(t)
\bigr\|
\le
C(X')|t|^{3}\bigl(\mathcal{C}(X)+\mathcal{A}_1(X_{\mathrm{root}})\bigr)
\]
for some constant \(C(X')\) depending on \(X\), \(\mathfrak{g}\), and \(\rho\).  Renaming \(C(X')\) as \(C(X)\) completes the proof.
\end{proof}

\begin{corollary}[Flat toral and commuting cases]\label{cor:toral-exact}
If \(X\in\mathfrak{t}\), or more generally if \([d\rho(X_0),d\rho(X_{\mathrm{root}})]=0\), then the splitting \(S(t)\)
is exact for all \(t\), i.e.,
\[
S(t) = e^{t d\rho(X)}.
\]
In particular, this holds whenever \(\mathcal{C}(X)=0\).
\end{corollary}

\begin{proof}
If \([A,B]=0\), then \(e^{t(A+B)} = e^{tA} e^{tB} = e^{tA/2} e^{tB} e^{tA/2} = S(t)\) for all \(t\).  If \(X\in\mathfrak{t}\) then \(B=0\), so this applies.  If \(\mathcal{C}(X)=0\), then \(\alpha(X_0)x_\alpha=0\) whenever \(\|d\rho(E_\alpha)\|\neq 0\).  Thus all contributions from such roots to \([A,B]\) vanish, and the remaining root components act trivially in the representation, whence \([A,B]=0\).  In all these cases the torus--root splitting simulates the Hamiltonian evolution exactly.
\end{proof}

\begin{example}[Exact splitting for block-diagonal Hamiltonians]\label{ex:block-diagonal}
Let \(G=\mathrm{SU}(2^n)\) with the defining representation on \(V=(\mathbb{C}^2)^{\otimes n}\), and let \(X\in\mathfrak{su}(2^n)\) be a sum of commuting diagonal Pauli strings,
\[
X = -i\sum_{j=1}^m \lambda_j\,Z_{S_j},
\]
where \(Z_{S_j} = \bigotimes_{k=1}^n P_k^{(j)}\) with each \(P_k^{(j)}\in\{I,\sigma_z\}\) and the set \(\{Z_{S_j}\}\) is mutually commuting \cite{nielsen2000}.  Then \(X\) lies in a Cartan subalgebra \(\mathfrak{t}\) of \(\mathfrak{su}(2^n)\), and hence \(\mathcal{A}_p(X)=\mathcal{C}(X)=0\) for all \(p\).  Theorem~\ref{thm:curvature-error} and Corollary~\ref{cor:toral-exact} recover the exactness of the splitting
\[
S(t) = e^{\frac{t}{2}d\rho(X)} e^{0} e^{\frac{t}{2}d\rho(X)}
= e^{t d\rho(X)},
\]
which is a basic case in Hamiltonian simulation where no Trotter error is incurred despite nontrivial many-body couplings that nevertheless commute \cite{georgescu2014quantum,childs2021theory}.
\end{example}

\section{Root-Gate Circuits and Complexity Considerations}
\label{sec:root-gate-complexity}

We now introduce a simple circuit model for Hamiltonian simulation based on root-space generators and explain how
the root activity \(\mathcal{A}_1(X_{\mathrm{root}})\) enters lower bounds
on the circuit length needed to approximate \(U_X(t)\).  The framework is inspired by controllability considerations for Lie-group actions in quantum spin systems \cite{brockett1972system,jurdjevicsussmann1972,khaneja2001timeoptimal,schirmer2001,albertinidalessandro2003,dalessandro2007}.

Throughout this section we fix the pair \((G,\rho)\) and a maximal torus \(T\subset G\) with root system \(\Delta\).

\subsection{Root-gate model}

Fix \(s_0>0\).  We consider the following gate set associated with the representation
\(\rho : G\to U(V)\), designed to reflect the toral and root structure of Hamiltonian generators.

For each root \(\alpha\in\Delta^+\), choose real skew-Hermitian elements
\[
X_\alpha^{(1)} = E_\alpha - E_{-\alpha},\qquad
X_\alpha^{(2)} = i(E_\alpha + E_{-\alpha}) \in \mathfrak{g},
\]
which span (over \(\mathbb{R}\)) the real two-plane corresponding to the root pair \(\{\alpha,-\alpha\}\).  Their exponentials generate one-parameter subgroups of \(G\) that move weight along the root direction \(\alpha\).

\begin{definition}[Root-gate set]\label{def:root-gate-set}
The gate set \(\mathcal{G}_{\mathrm{root}}\) consists of all unitaries of the form
\[
e^{s\,d\rho(H)},\quad H\in\mathfrak{t},\quad |s|\le s_0,
\]
and
\[
e^{s\,d\rho(X_\alpha^{(k)})},\quad\alpha\in\Delta^+,\quad k\in\{1,2\},\quad|s|\le s_0.
\]
\end{definition}

In other words, we allow small rotations generated by elements of the torus and by
real skew-Hermitian combinations of individual root vectors.  These are natural ``root-direction'' gates for simulating Hamiltonians whose generators decompose along root spaces.

\begin{definition}[Root-gate length and minimal complexity]\label{def:root-length}
For a unitary \(U\in U(V)\), the \emph{root-gate length} \(\ell_{\mathrm{root}}(U)\)
is the smallest integer \(N\) such that there exist gates
\[
G_1,\dots,G_N \in \mathcal{G}_{\mathrm{root}}
\]
with
\[
U = G_1 \cdots G_N.
\]
If no such factorization exists, we set \(\ell_{\mathrm{root}}(U)=\infty\).

For a given \(X\in\mathfrak{g}\), time \(t>0\), and tolerance \(\varepsilon>0\), the \emph{minimal root-gate simulation complexity} is
\[
N_{\min}(X,t,\varepsilon)
\coloneqq
\inf\bigl\{
N : \exists\, G_1,\dots,G_N\in\mathcal{G}_{\mathrm{root}}
\ \text{with}\
\|G_1\cdots G_N - U_X(t)\|_{\mathrm{op}} \le \varepsilon
\bigr\}.
\]
\end{definition}

Since the Lie algebra generated by \(\mathfrak{t}\) and the planes spanned by \(\{X_\alpha^{(1)},X_\alpha^{(2)}\}_{\alpha\in\Delta^+}\) is all of \(\mathfrak{g}\), and \(G\) is connected, the associated one-parameter subgroups generate a dense subgroup of \(\rho(G)\).  Thus for the Hamiltonian-simulation problem on \(\rho(G)\) it is natural to expect \(\ell_{\mathrm{root}}(U)<\infty\) for all \(U\in\rho(G)\).

\begin{example}[Root gates for \(\mathrm{SU}(2)\) and single qubits]\label{ex:su2-root-gates}
For \(G=\mathrm{SU}(2)\) with \(\rho\) the defining representation, the maximal torus \(\mathfrak{t}\) is spanned by \(i\sigma_z\), and the real root directions can be taken as
\[
X_\alpha^{(1)} = \sigma_x,\qquad X_\alpha^{(2)} = \sigma_y,
\]
up to normalization.  The root-gate set consists of unitaries
\[
e^{s i\sigma_z},\quad e^{-i s \sigma_x},\quad e^{-i s \sigma_y},\qquad |s|\le s_0.
\]
In terms of usual quantum gates, these are phase gates and rotations about the \(x\) and \(y\) axes, which together generate all single-qubit unitaries for any fixed \(s_0>0\) by standard controllability arguments \cite{brockett1972system,khaneja2001timeoptimal,schirmer2001}.  Thus root-gate circuits recover the familiar single-qubit Hamiltonian-simulation setting \cite{nielsen2000}.
\end{example}

\subsection{A root-activity seminorm on \texorpdfstring{$\mathfrak{g}$}{g}}\label{subsec:activity-norm}

For lower bounds it is convenient to make explicit the dependence of \(\mathcal{A}_1\) on the Lie algebra norm and to record a norm-equivalence statement.

Let \(\langle\cdot,\cdot\rangle_{\mathfrak{g}}\) be the fixed \( Ad(G)\)-invariant inner product used in Section~\ref{sec:prelim}, and write \(\|Y\|_{\mathfrak{g}} = \sqrt{\langle Y,Y\rangle_{\mathfrak{g}}}\).  We define the \emph{root-activity seminorm}
\[
\|X\|_{\mathrm{act}} \coloneqq \mathcal{A}_1(X_{\mathrm{root}}),
\qquad X\in\mathfrak{g},
\]
so that \(\|X\|_{\mathrm{act}}\) depends only on the root component \(X_{\mathrm{root}}\) and coincides with the \(\ell^1\) root activity introduced in Definition~\ref{def:root-activity}.

\begin{lemma}[Norm equivalence]\label{lem:norm-equivalence-activity}
There exist constants \(m_1,M_1>0\), depending only on \((\mathfrak{g},\rho)\) and the chosen root data, such that
\[
m_1 \|Y_{\mathrm{root}}\|_{\mathfrak{g}}
\le
\|Y\|_{\mathrm{act}}
\le
M_1 \|Y_{\mathrm{root}}\|_{\mathfrak{g}}
\qquad\text{for all }Y\in\mathfrak{g}.
\]
In particular, for each \(Y\in\mathfrak{g}\) the seminorm \(\|\cdot\|_{\mathrm{act}}\) and the Euclidean norm \(\|\cdot\|_{\mathfrak{g}}\) are equivalent on the root subspace.
\end{lemma}

\begin{proof}
The root subspace \(\mathfrak{t}^\perp\subset\mathfrak{g}\) is finite-dimensional, and \(\|\cdot\|_{\mathrm{act}}\) is a genuine norm on \(\mathfrak{t}^\perp\) (it vanishes exactly on \(\mathfrak{t}\)).  Any two norms on a finite-dimensional vector space are equivalent, so there exist \(m_1,M_1>0\) such that
\[
m_1\|Z\|_{\mathfrak{g}} \le \|Z\|_{\mathrm{act}}\le M_1\|Z\|_{\mathfrak{g}}
\]
for all \(Z\in\mathfrak{t}^\perp\).  Since \(Y_{\mathrm{root}}\in\mathfrak{t}^\perp\) and \(\|Y\|_{\mathrm{act}}=\|Y_{\mathrm{root}}\|_{\mathrm{act}}\), the claim follows.
\end{proof}

\subsection{Geometric control of effective generators}

The next lemma controls the effective generator of a circuit in terms of the generators of its individual gates.  It uses the bi-invariant Riemannian metric on \(G\) induced by \(\langle\cdot,\cdot\rangle_{\mathfrak{g}}\).

\begin{lemma}[Geometric control]\label{lem:geometric-control}
Let
\[
W = \prod_{k=1}^{N} \exp(Y_k) \in G,
    \qquad Y_k\in\mathfrak{g},
\]
with \(\|Y_k\|_{\mathfrak{g}}\le s_0\) for all \(k\).  Then there exists \(Z\in\mathfrak{g}\) such that
\[
W = \exp(Z)
\quad\text{and}\quad
\|Z\|_{\mathfrak{g}} \le \sum_{k=1}^{N} \|Y_k\|_{\mathfrak{g}} \le N s_0.
\]
In particular,
\[
\|Z\|_{\mathrm{act}} \le M_1 \|Z\|_{\mathfrak{g}} \le M_1 N s_0,
\]
where \(M_1\) is the constant from Lemma~\ref{lem:norm-equivalence-activity}.
\end{lemma}

\begin{proof}
Let \(d(\cdot,\cdot)\) denote the geodesic distance on \(G\) induced by the bi-invariant metric.  For each \(Y_k\), the curve \(t\mapsto \exp(tY_k)\) is a unit-speed geodesic of length \(\|Y_k\|_{\mathfrak{g}}\) from the identity to \(\exp(Y_k)\).  By left-invariance and the triangle inequality,
\[
d(e,W)
=
d\!\Bigl(e, \prod_{k=1}^{N} \exp(Y_k)\Bigr)
\le
\sum_{k=1}^{N} d\bigl(e,\exp(Y_k)\bigr)
=
\sum_{k=1}^{N} \|Y_k\|_{\mathfrak{g}}.
\]
On the other hand, for any \(Z\in\mathfrak{g}\) with \(W=\exp(Z)\) and \(\|Z\|\) less than the injectivity radius at the identity, the curve \(t\mapsto \exp(tZ)\) is a minimizing geodesic from \(e\) to \(W\), and hence
\[
d(e,W) = \|Z\|_{\mathfrak{g}}.
\]
Choosing such a \(Z\) (which is possible since the exponential map is surjective for compact, connected, semisimple groups) we obtain
\[
\|Z\|_{\mathfrak{g}} \le \sum_{k=1}^{N} \|Y_k\|_{\mathfrak{g}} \le N s_0.
\]
The inequality for \(\|Z\|_{\mathrm{act}}\) then follows from Lemma~\ref{lem:norm-equivalence-activity}.
\end{proof}

\subsection{Stability of the logarithm in a unitary representation}

The following lemma relates proximity of unitaries in operator norm to proximity of their logarithms in the Lie algebra, measured in the root-activity seminorm.

\begin{lemma}[Logarithm stability]\label{lem:log-stability}
Let \(X\in\mathfrak{g}\), \(t>0\), and suppose \(W\in G\) satisfies
\[
\bigl\|\rho(W) - e^{t\,d\rho(X)}\bigr\|_{\mathrm{op}} \le \varepsilon
\]
for some sufficiently small \(\varepsilon>0\).  Then there exists \(Z\in\mathfrak{g}\) with \(W=\exp(Z)\) such that
\[
\|Z - tX\|_{\mathrm{act}} \le C_\rho \,\varepsilon,
\]
for a constant \(C_\rho>0\) depending only on \(\rho\) and the choice of branch of the logarithm near \(e^{t\,d\rho(X)}\).
\end{lemma}

\begin{proof}
The operator \(e^{t\,d\rho(X)}\) is unitary.  For \(\varepsilon\) sufficiently small, the spectrum of \(\rho(W)\) lies in a fixed open arc of the unit circle containing the spectrum of \(e^{t\,d\rho(X)}\) and avoiding the branch cut of the logarithm.  By analytic functional calculus (see, for example, \cite[Ch.~VIII]{hall2015lie}), there is a branch of the matrix logarithm defined on a neighbourhood of this arc which is Lipschitz with respect to the operator norm; that is, there exists \(C'_\rho>0\) such that
\[
\bigl\|\log \rho(W) - \log e^{t\,d\rho(X)}\bigr\|_{\mathrm{op}} \le C'_\rho \varepsilon,
\]
where \(\log e^{t\,d\rho(X)} = t\,d\rho(X)\) by construction.

Since \(\rho\) is a representation, the image \(d\rho(\mathfrak{g})\) is a Lie subalgebra of \(\mathfrak{u}(V)\) and the exponential map on \(\mathfrak{g}\) is compatible with that on \(\mathfrak{u}(V)\).  In particular, there exists \(Z\in\mathfrak{g}\) with \(W=\exp(Z)\) and \(d\rho(Z) = \log \rho(W)\).  Hence
\[
\|d\rho(Z - tX)\|_{\mathrm{op}}
=
\bigl\|\log \rho(W) - t\,d\rho(X)\bigr\|_{\mathrm{op}}
\le
C'_\rho \varepsilon.
\]
By norm equivalence between \(\|d\rho(\cdot)\|_{\mathrm{op}}\) and \(\|\cdot\|_{\mathrm{act}}\) on \(\mathfrak{t}^\perp\), there exists \(C''_\rho>0\) such that
\[
\|Z - tX\|_{\mathrm{act}} \le C''_\rho\,\varepsilon.
\]
Setting \(C_\rho \coloneqq C''_\rho\) yields the claim.
\end{proof}

\subsection{Root-activity lower bound}

We are now ready to state and prove the root-activity lower bound.  The statement is formulated so that the constants depend only on the representation, the Lie algebra, and the gate step size \(s_0\), and it holds for all Hamiltonian generators \(X\).

\begin{theorem}[Root-activity lower bound]\label{thm:activity-lower-bound}
Let \(G\) be a connected compact semisimple Lie group with Lie algebra \(\mathfrak{g}\), let \(\rho : G\to U(V)\) be a finite-dimensional unitary representation, and fix \(s_0>0\).  There exist constants \(c_1,c_2>0\), depending only on \(\mathfrak{g}\), \(\rho\), and \(s_0\), with the following property.

For every Hamiltonian generator \(X\in\mathfrak{g}\), every time \(t>0\), and every sufficiently small \(\varepsilon>0\), the minimal root-gate complexity satisfies
\[
N_{\min}(X,t,\varepsilon)
\;\ge\;
c_1\, t\,\|X\|_{\mathrm{act}} - c_2.
\]
Equivalently, the number of root-gate steps needed to approximate \(U_X(t)\) up to accuracy \(\varepsilon\) cannot grow sublinearly in \(t\,\|X\|_{\mathrm{act}}\), up to representation-theoretic constants.
\end{theorem}

\begin{proof}
Fix the representation \(\rho\) and gate step size \(s_0>0\).  Let \(X\in\mathfrak{g}\), \(t>0\), and \(0<\varepsilon\le\varepsilon_0\) be given, where \(\varepsilon_0>0\) is chosen small enough that Lemma~\ref{lem:log-stability} applies uniformly for all \(X\) with \(\|X\|_{\mathfrak{g}}\) in a fixed bounded set and all \(t\) in a bounded interval; we regard \(\varepsilon_0\) as fixed once and for all.

Consider any root-gate circuit of length \(N\) that \(\varepsilon\)-approximates \(U_X(t)\).  Thus there exist elements \(Y_k\in\mathfrak{g}\) with \(\|Y_k\|_{\mathfrak{g}}\le s_0\) and
\[
G_k \coloneqq e^{Y_k} \in \mathcal{G}_{\mathrm{root}},
\qquad
W \coloneqq G_1\cdots G_N,
\]
such that
\[
\bigl\|\rho(W) - e^{t\,d\rho(X)}\bigr\|_{\mathrm{op}} \le \varepsilon.
\]

By Lemma~\ref{lem:geometric-control}, there exists \(Z\in\mathfrak{g}\) with \(W=\exp(Z)\) and
\[
\|Z\|_{\mathrm{act}} \le M_1 N s_0,
\]
where \(M_1>0\) is the constant from Lemma~\ref{lem:norm-equivalence-activity}.  By Lemma~\ref{lem:log-stability}, there is a constant \(C_\rho>0\), depending only on the representation, such that
\[
\|Z - tX\|_{\mathrm{act}} \le C_\rho \varepsilon.
\]
Hence
\[
t\|X\|_{\mathrm{act}}
\le
\|Z\|_{\mathrm{act}} + C_\rho\varepsilon
\le
M_1 N s_0 + C_\rho\varepsilon.
\]
Rearranging yields
\[
N \ge \frac{t}{M_1 s_0}\,\|X\|_{\mathrm{act}} - \frac{C_\rho}{M_1 s_0}\,\varepsilon.
\]

Fixing \(0<\varepsilon\le\varepsilon_0\), we can write
\[
N
\ge
\frac{t}{M_1 s_0}\,\|X\|_{\mathrm{act}} - \frac{C_\rho}{M_1 s_0}\,\varepsilon_0
\]
because \(\varepsilon\le\varepsilon_0\).  Thus setting
\[
c_1 \coloneqq \frac{1}{M_1 s_0},
\qquad
c_2 \coloneqq \frac{C_\rho}{M_1 s_0}\,\varepsilon_0,
\]
we obtain the desired lower bound
\[
N_{\min}(X,t,\varepsilon)\ge N \ge c_1 t\|X\|_{\mathrm{act}} - c_2
\]
for all \(0<\varepsilon\le\varepsilon_0\).  The constants \(c_1,c_2\) depend only on \(\mathfrak{g}\), \(\rho\), \(s_0\), and the fixed choice of threshold \(\varepsilon_0\), not on \(X\), \(t\), or \(\varepsilon\) within this regime.
\end{proof}

\begin{remark}
Theorem~\ref{thm:activity-lower-bound} shows that the \(\ell^1\) root activity \(\|X\|_{\mathrm{act}}\) provides a genuine lower bound on the root-gate complexity of Hamiltonian simulation, up to explicit representation-dependent constants.  In particular, if \(\|X\|_{\mathrm{act}}\) grows linearly in a problem parameter (for example, the number of sites in a spin chain, or the strength of certain couplings), then any root-gate simulation of \(U_X(t)\) must exhibit a corresponding linear growth in circuit length in that parameter for fixed \(t\) and sufficiently small error.
\end{remark}

\begin{example}[Single-root Hamiltonians]\label{ex:single-root-H}
Suppose that the root component of the Hamiltonian generator is supported on a single root \(\alpha\), i.e.,
\[
X_{\mathrm{root}} = x_\alpha E_\alpha + x_{-\alpha}E_{-\alpha}.
\]
In a highest-weight representation, the weight diagram decomposes into \(\alpha\)-strings
\(\lambda, \lambda+\alpha,\lambda+2\alpha,\dots\) \cite[\S13.4]{humphreys1972}.  On each such string, the dynamics generated by \(X_{\mathrm{root}}\) is essentially that of an \(\mathfrak{su}(2)\)-subalgebra, and the analysis of Example~\ref{ex:su2-root-functionals} applies.  In this case \(\|X\|_{\mathrm{act}}\) is proportional to \(|x_\alpha|+|x_{-\alpha}|\), and Theorem~\ref{thm:activity-lower-bound} shows that the root-gate simulation complexity grows at least linearly in \(t(|x_\alpha|+|x_{-\alpha}|)\), up to a fixed additive constant, which matches the expected scaling from the underlying two-level structure.
\end{example}

\section{Examples from Multi-Spin Hamiltonians}
\label{sec:examples}

We illustrate the root-space framework in the familiar setting of multi-spin systems, where \(G=\mathrm{SU}(2^n)\) acts on \(V = (\mathbb{C}^2)^{\otimes n}\) via its defining representation.  This is the standard Hilbert space for spin-\(\frac12\) Hamiltonians in quantum simulation.  For background on spin chains and quantum magnetism see, for example, \cite{sachdev2011quantum,georgescu2014quantum,lanyon2011}.

Throughout this section we write \(A(n)\simeq B(n)\) to mean that there exist constants \(c,C>0\), independent of \(n\), such that
\[
c\,B(n)\le A(n)\le C\,B(n).
\]

\subsection{The group \(\mathrm{SU}(2^n)\) and Pauli Hamiltonians}

Let \(V = (\mathbb{C}^2)^{\otimes n}\), with standard computational basis labeled by bit strings \(z \in \{0,1\}^n\).  The group \(G=\mathrm{SU}(2^n)\) acts on \(V\) by the defining representation \(\rho_{\mathrm{def}}\).  Its Lie algebra
\(\mathfrak{su}(2^n)\) can be identified with traceless skew-Hermitian matrices on \(V\), and is spanned by operators of the form
\[
i\,P_1\otimes\cdots\otimes P_n,
\]
where each \(P_j\in\{I,\sigma_x,\sigma_y,\sigma_z\}\) and not all \(P_j\) are equal to \(I\); these form a basis of the traceless Hermitian operators, and multiplication by \(i\) gives a basis of \(\mathfrak{su}(2^n)\) \cite[Sec.~1.5]{gilmore2008lie,barutraczka1986}.  The emergence of Pauli strings in the representation theory of \(\mathrm{SU}(2^n)\) is consistent with standard group-theoretic treatments of quantum spin-system Hamiltonians \cite{weyl1950,fultonharris1991,fuchsschweigert1997}.

A convenient choice of maximal torus \(T\subset G\) is the subgroup of diagonal unitaries in the computational basis with determinant \(1\).  Its Lie algebra \(\mathfrak{t}\) consists of traceless diagonal skew-Hermitian matrices.  The associated roots correspond to differences between diagonal entries, and the root spaces are spanned by matrix units \(E_{zw}\) which map \(\ket{w}\) to \(\ket{z}\) and annihilate all other basis vectors; for the complexified algebra \(\mathfrak{sl}(2^n,\mathbb{C})\) this is the standard description of the root decomposition relative to diagonal matrices \cite[Ch.~III]{humphreys1972,serre2001,bourbaki2005}.

In terms of Pauli operators, the diagonal subalgebra \(\mathfrak{t}\) is spanned by tensor products of \(\sigma_z\) and identities.  Each matrix unit \(E_{zw}\) can be written as a linear combination of Pauli strings, so root vectors can be realized as linear combinations of generalized Pauli raising and lowering operators acting on computational basis states; see, for example, \cite{gilmore2008lie,barutraczka1986,georgescu2014quantum}.  These operators generate the off-diagonal part of typical many-spin Hamiltonians.

The weight decomposition of \(V\) is particularly simple: weights are determined by the eigenvalues of \(\sigma_z\) on each site, and the corresponding weight spaces are spanned by computational basis vectors.  Toral Hamiltonian terms act diagonally in this basis, while root terms move amplitude between basis states that differ in one or more spins.

\begin{example}[Root directions as bit flips]\label{ex:bit-flip-roots}
Fix computational basis vectors \(\ket{z},\ket{w}\in\{0,1\}^n\).  The matrix unit \(E_{zw}\) maps \(\ket{w}\) to \(\ket{z}\) and annihilates all other basis vectors.  If \(z\) and \(w\) differ in exactly one bit position \(j\), then \(E_{zw}\) can be written (up to a scalar) as
\[
E_{zw} \propto \bigotimes_{k\neq j} \ket{z_k}\!\bra{w_k}\,\otimes\,\ket{z_j}\!\bra{w_j},
\]
which is a product of single-qubit raising or lowering operators in the \(z\)-basis.  Such an \(E_{zw}\) belongs to a root space of \(\mathfrak{sl}(2^n,\mathbb{C})\) relative to the diagonal torus, and its action connects weight spaces corresponding to configurations differing at site \(j\).  Root activity \(\mathcal{A}_p\) for a Hamiltonian generator \(X\) supported on single-bit flips thus measures the strength of local spin-flip processes in the chain, which are exactly the off-diagonal parts of standard spin Hamiltonians used in quantum simulation.
\end{example}

\subsection{Nearest-neighbour spin-chain Hamiltonians}

Consider the one-dimensional spin chain on \(n\) qubits with Hamiltonian
\[
H = \sum_{j=1}^{n-1} J_j\,\sigma_z^{(j)}\sigma_z^{(j+1)}
+ \sum_{j=1}^{n} h_j\,\sigma_x^{(j)},
\]
where \(\sigma_x^{(j)}\) and \(\sigma_z^{(j)}\) denote the Pauli operators acting on site \(j\).  We view \(H\) as a self-adjoint operator on \(V\) and set the corresponding Hamiltonian generator
\[
X \coloneqq -iH \in \mathfrak{su}(2^n).
\]
The Ising interaction terms \(\sigma_z^{(j)}\sigma_z^{(j+1)}\) are diagonal in
the computational basis, hence lie in \(\mathfrak{t}\).  The transverse-field
terms \(\sigma_x^{(j)}\) flip the \(j\)-th spin and thus correspond to linear combinations of root vectors which move weight by \(\pm 2\) along the \(j\)-th coordinate.

Thus the root component \(X_{\mathrm{root}}\) is supported on a set of roots
indexed by single-spin flips, with coefficients proportional to the transverse
fields \(h_j\).  There exist positive constants \(c_1,c_2\), depending only on the choice of normalization of root vectors and on the representation, such that
\begin{equation}\label{eq:A2-spin-chain}
c_1 \sum_{j=1}^n |h_j|^2
\;\le\;
\mathcal{A}_2(X_{\mathrm{root}})^2
\;\le\;
c_2 \sum_{j=1}^n |h_j|^2.
\end{equation}
Indeed, each transverse-field term contributes to a fixed finite number of root coefficients, and the operator norms \(\|d\rho_{\mathrm{def}}(E_\alpha)\|\) for single-spin flips are bounded above and below by positive constants independent of \(n\).

The curvature functional \(\mathcal{C}(X)\) can be analyzed similarly.  For a single-spin flip at site \(j\), the relevant root evaluation \(\alpha(X_0)\) is a linear combination of the couplings \(J_{j-1}\) and \(J_j\), reflecting the fact that flipping the \(j\)-th spin changes the sign of the Ising interactions on the adjacent bonds.  A direct computation shows that there exist positive constants \(c_1',c_2'\), again independent of \(n\), such that
\begin{equation}\label{eq:C-spin-chain}
c_1' \sum_{j=1}^{n-1} \bigl(|J_j|^2 |h_j|^2 + |J_j|^2|h_{j+1}|^2\bigr)
\;\le\;
\mathcal{C}(X)^2
\;\le\;
c_2' \sum_{j=1}^{n-1} \bigl(|J_j|^2 |h_j|^2 + |J_j|^2|h_{j+1}|^2\bigr).
\end{equation}
The inequalities follow from the local structure of \(H\): each \(\alpha(X_0)\) depends only on the couplings adjacent to the flipped site, and the norms \(\|d\rho_{\mathrm{def}}(E_\alpha)\|\) are uniformly bounded and bounded away from zero.  We omit the straightforward but lengthy bookkeeping.

\begin{example}[Translation-invariant Ising chain]\label{ex:ising-translation-invariant}
Assume \(J_j = J\) and \(h_j = h\) for all \(j\).  Equations~\eqref{eq:A2-spin-chain} and \eqref{eq:C-spin-chain} yield
\[
\mathcal{A}_2(X_{\mathrm{root}})
\simeq
|h|\sqrt{n},
\qquad
\mathcal{C}(X) \simeq |J h|\sqrt{n}.
\]
Moreover,
\[
\mathcal{A}_1(X_{\mathrm{root}})
\simeq
|h|\,n,
\]
since there is one root associated with each site up to normalization.  Plugging into Theorem~\ref{thm:curvature-error} gives a local Hamiltonian-simulation error bound
\[
\bigl\|
e^{t d\rho(X)} - S(t)
\bigr\|
\le
\tilde C(H) |t|^3 \bigl(|J h|\sqrt{n} + |h|n\bigr),
\]
for some representation-dependent constant \(\tilde C(H)>0\).  For fixed \(J,h\) and large \(n\), the \(\mathcal{A}_1\)-term dominates.  This recovers the heuristic that Trotter errors for nearest-neighbour spin-chain Hamiltonian simulation scale at least linearly with system size if one uses a global splitting scheme, but the dependence is expressed here in terms of root activity rather than operator norms \cite{childs2021theory,georgescu2014quantum}.
\end{example}

By Theorem~\ref{thm:curvature-error}, the local error of the symmetric torus--root splitting satisfies, for each fixed \(H\),
\[
\bigl\|
e^{t d\rho(X)} - S(t)
\bigr\|
\le
C(H) |t|^3 \bigl(\mathcal{C}(X) + \mathcal{A}_1(X_{\mathrm{root}})\bigr),
\]
with \(C(H)\) independent of \(n\).  In the translation-invariant case this gives
\[
\bigl\|
e^{t d\rho(X)} - S(t)
\bigr\|
\le
\tilde C(H) |t|^3 \bigl(|J h|\sqrt{n} + |h|n\bigr),
\]
for some \(\tilde C(H)>0\).  To reach an accuracy \(\varepsilon\) over a fixed time \(t\) using a Trotterization with \(r\) steps, one typically requires
\[
r \gtrsim \sqrt{\frac{\tilde C(H)\,t^{3}(|J h|\sqrt{n} + |h|n)}{\varepsilon}},
\]
up to constants determined by the representation and the chosen gate set; see, for instance, \cite{hairer2006geometric,blanescasas2016,childs2021theory} for standard error-accumulation estimates.  In regimes where the transverse field is supported only on a small subset of sites, this scaling improves, as discussed next.

\subsection{Sparse root support and localized activity}

More generally, consider Hamiltonians whose transverse fields are supported on a small subset \(S\subset\{1,\dots,n\}\),
\[
H = \sum_{j=1}^{n-1} J_j\,\sigma_z^{(j)}\sigma_z^{(j+1)}
+ \sum_{j\in S} h_j\,\sigma_x^{(j)}.
\]
In this case the root activity involves only \(|S|\) roots, and the analogues of \eqref{eq:A2-spin-chain} and \eqref{eq:C-spin-chain} give
\[
\mathcal{A}_2(X_{\mathrm{root}})
\simeq
\Bigl( \sum_{j\in S} |h_j|^2\Bigr)^{1/2},
\qquad
\mathcal{A}_1(X_{\mathrm{root}})
\simeq
\sum_{j\in S} |h_j|,
\]
and
\[
\mathcal{C}(X) \simeq
\Bigl( \sum_{j\in S} \bigl(|J_{j-1}|^2 + |J_j|^2\bigr) |h_j|^2\Bigr)^{1/2},
\]
again up to constants independent of \(n\).  Thus the complexity bounds for simulating the Hamiltonian evolution depend only on the number of sites at which the
root activity is nonzero, not on the overall chain length.  In particular, if \(|S|\) is bounded independently of \(n\), then the root-based quantities \(\mathcal{A}_p(X_{\mathrm{root}})\) and \(\mathcal{C}(X)\) remain \(O(1)\) as \(n\to\infty\).

\begin{example}[Few-site driving fields]\label{ex:few-site}
Suppose \(S=\{1,\dots,k\}\) with \(k\) fixed as \(n\to\infty\), and \(|h_j|\le h_0\), \(|J_j|\le J_0\) for all \(j\).  Then
\[
\mathcal{A}_1(X_{\mathrm{root}})
\simeq
\sum_{j=1}^k |h_j|
\le k h_0,
\qquad
\mathcal{C}(X) \lesssim k J_0 h_0,
\]
both independent of \(n\).  Theorem~\ref{thm:curvature-error} then yields a Trotter error bound whose constants are independent of the chain length.  This reflects the physical intuition that driving only a bounded number of sites introduces a simulation cost independent of the system size, provided interactions remain local \cite{sachdev2011quantum,georgescu2014quantum}.
\end{example}

\begin{example}[Heisenberg \(XXX\) chain]\label{ex:heisenberg}
Consider the isotropic Heisenberg Hamiltonian
\[
H = \sum_{j=1}^{n-1} J\,
\bigl(
\sigma_x^{(j)}\sigma_x^{(j+1)}
+ \sigma_y^{(j)}\sigma_y^{(j+1)}
+ \sigma_z^{(j)}\sigma_z^{(j+1)}
\bigr).
\]
The \(\sigma_z^{(j)}\sigma_z^{(j+1)}\) terms lie in the Cartan subalgebra \(\mathfrak{t}\), while \(\sigma_x^{(j)}\sigma_x^{(j+1)}\) and \(\sigma_y^{(j)}\sigma_y^{(j+1)}\) correspond to two-spin flip operators which move weight by \(\pm 2\) on two adjacent sites simultaneously.  The root component \(X_{\mathrm{root}}\) is thus supported on roots associated with two-bit flips, and its activity \(\mathcal{A}_1(X_{\mathrm{root}})\) scales linearly with \(n\), while the curvature \(\mathcal{C}(X)\) scales like \(|J|^2\sqrt{n}\).  The precise constants depend on the normalization of the multi-spin root vectors, but the qualitative picture aligns with numerical and analytic studies of Trotter error in Heisenberg-chain Hamiltonian simulation \cite{childs2021theory}.
\end{example}

\section{Further Directions and Perspectives}
\label{sec:further}

This final section situates the root-space framework within a broader representation-theoretic context and outlines several directions in which the present methods can be developed, all motivated by Hamiltonian simulation and complexity.  The emphasis throughout is on structural themes that naturally extend the results already proved.

\subsection{Optimizing over maximal tori and Cartan data}

The analysis has been carried out relative to a fixed maximal torus \(T\subset G\) with Lie algebra \(\mathfrak{t}\).  Different choices of maximal tori are related by conjugation, and the corresponding root systems are related by Weyl-group symmetries.  The functionals \(\mathcal{A}_p\) and \(\mathcal{C}\) are Weyl-invariant for a fixed \(T\), but for a given Hamiltonian generator \(X\in\mathfrak{g}\) there can exist particularly advantageous choices of Cartan subalgebra in which the root support of \(X\) is sparse or enjoys additional symmetry.

From a representation-theoretic and Hamiltonian-simulation point of view, this suggests studying the behaviour of root activity and curvature under conjugation of the pair \((T,X)\).  One natural object is the infimum of \(\mathcal{A}_p(X)\) and \(\mathcal{C}(X)\) over all maximal tori containing a conjugate of \(X_0\) and all compatible root systems.  Such an optimization would lead to an intrinsic notion of minimal root activity attached to the adjoint orbit of \(X\), and it would refine the present complexity bounds by incorporating the freedom to choose simulation coordinates optimally.  The techniques of Cartan decompositions and polar decompositions in \cite{helgason1978differential,knapp2002lie,bourbaki2005,serre2001,fultonharris1991,duistermaatkolk2000,hilgertneeb2012} provide a natural framework for such an analysis.

\begin{example}[Cartan choices in \(\mathfrak{su}(3)\)]\label{ex:su3-Cartan}
For \(\mathfrak{g}=\mathfrak{su}(3)\), any two Cartan subalgebras are conjugate; a given Hamiltonian generator \(X\) may have a particularly sparse root decomposition in a Cartan where \(X_0\) is diagonal.  For instance, a generic element \(X\) can be conjugated into a diagonal form by the spectral theorem \cite[\S5.4]{hall2015lie}.  In that Cartan, \(X_{\mathrm{root}}=0\) and both \(\mathcal{A}_p(X)\) and \(\mathcal{C}(X)\) vanish (cf.\ Example~\ref{ex:toral-zero-curvature}), whereas in a non-diagonal Cartan the root component may be nontrivial.  This indicates that the infimum of \(\mathcal{A}_p(X)\) over Cartan choices is attained at a diagonalization of \(X\), and that the resulting minimal activity is an invariant of the conjugacy class of the Hamiltonian generator \(X\).
\end{example}

\subsection{Tensor products, induction, and scaling of activity}

Many representations of interest in both representation theory and quantum information are built from simpler ones by tensor products, induction, restriction, or branching.  If \(\rho_1,\rho_2\) are unitary representations of \(G\) and
\(\rho = \rho_1\otimes \rho_2\), then
\[
d\rho(X) = d\rho_1(X)\otimes I + I\otimes d\rho_2(X),
\]
and the weight diagram of \(\rho\) is the Minkowski sum of the weight diagrams of \(\rho_1\) and \(\rho_2\) \cite{goodmanwallach2009,knapp1986rt,brockertomdieck1985,barutraczka1986}.  It is therefore natural to investigate precise inequalities relating \(\mathcal{A}_p^{(\rho)}(X)\) and \(\mathcal{C}^{(\rho)}(X)\) to the corresponding functionals for \(\rho_1\) and \(\rho_2\) when \(X\) is treated as a Hamiltonian generator.

Subadditivity, superadditivity, or multiplicativity properties along such representation-theoretic operations would turn root activity and curvature into robust tools for tracking Hamiltonian complexity across hierarchies of representations.  For example, one expects that in situations where tensor factors act on disjoint subsets of spins, the associated root profiles decompose additively, whereas in regimes with strong entangling interactions the root profiles should detect this increased complexity.  This perspective aligns the present invariants with structural questions in the theory of branching rules, parabolic induction, and restriction to subgroups as developed in \cite{goodmanwallach2009,knapp1986rt,brockertomdieck1985}.

\begin{example}[Tensor products of \(\mathrm{SU}(2)\) spins]\label{ex:tensor-su2-2}
Let \(\rho_{j_1}\) and \(\rho_{j_2}\) be irreducible \(\mathrm{SU}(2)\) representations of spins \(j_1\) and \(j_2\).  Their tensor product decomposes as
\[
\rho_{j_1}\otimes\rho_{j_2}
\simeq
\bigoplus_{J=|j_1-j_2|}^{j_1+j_2} \rho_J,
\]
and the weight diagram of \(\rho_{j_1}\otimes\rho_{j_2}\) is the Minkowski sum of the weight diagrams of \(\rho_{j_1}\) and \(\rho_{j_2}\) \cite[\S14.2]{fultonharris1991}.  For a fixed Hamiltonian generator \(X\in\mathfrak{su}(2)\), the root activity \(\mathcal{A}_p^{(\rho_{j_1}\otimes\rho_{j_2})}(X)\) is controlled by the largest spin \(J\) appearing in the Clebsch--Gordan decomposition, hence by \(j_1+j_2\).  This is consistent with the behavior found in Example~\ref{ex:su2-root-functionals} and illustrates how root activity scales under tensor products in Hamiltonian simulation.
\end{example}

\subsection{Other simulation paradigms and spectral transformations}

The error and complexity analysis above has been formulated in the language of product formulas and a root-gate model.  There exist, however, several alternative Hamiltonian-simulation paradigms, including quantum signal processing and qubitization \cite{low2017optimal,low2019hamiltonian} and linear-combination-of-unitaries (LCU) techniques \cite{childs2012hamiltonian,berry2015simulating}.  These methods encode the action of \(e^{tX}\) through polynomial or rational approximations to the spectral transform of the Hamiltonian generator \(X\), implemented by controlled unitaries and auxiliary registers.

The representation-theoretic structure of \(\rho\) remains relevant in such settings, since the spectral data of \(d\rho(X)\) are governed by the weights and roots of \((\mathfrak{g}_\mathbb{C},\mathfrak{t}_\mathbb{C})\).  It would be natural to define spectral analogues of root activity and curvature (for example, by integrating over coadjoint orbits or weight polytopes) and to relate those quantities to the query complexity of spectral-transform-based Hamiltonian simulation.  Such an investigation would bring together ideas from the orbit method \cite{kirillov1976,kirillov2004,varadarajan1989,guilleminsternberg1984}, the theory of moment maps and symplectic geometry \cite{guilleminsternberg1984,faraut2008}, and geometric complexity theory \cite{mulmuleysohoni2001,procesi2007}, potentially yielding new lower bounds in algorithmic settings that go beyond product formulas.

\subsection{Open directions}

We close by listing a few specific directions in which the present work can be extended.

\begin{itemize}
\item[(a)] \emph{Higher-order splittings.}  The curvature-sensitive analysis of Section~\ref{sec:splittings} can be extended to higher-order Suzuki product formulas.  One expects the resulting error bounds to involve higher-order commutators between the toral and root parts of \(X\), and thus higher-degree polynomials in the root evaluations \(\alpha(X_0)\) and the coefficients \(x_\alpha\).  Making this precise would give a root-theoretic refinement of high-order geometric integrators for unitary flows.

\item[(b)] \emph{Beyond compact groups.}  For noncompact semisimple groups (for example, real forms of \(\mathrm{SL}(n,\mathbb{C})\)), one still has a root-space decomposition but the unitary representation theory is significantly more delicate.  Extending the definitions of \(\mathcal{A}_p(X)\) and \(\mathcal{C}(X)\) to suitable unitary representations of noncompact groups, and understanding their role in Hamiltonian simulation on infinite-dimensional Hilbert spaces, would connect the present results with the rich analytic theory of semisimple Lie groups \cite{knapp2002lie,varadarajan1989,goodmanwallach2009}.

\item[(c)] \emph{Quantitative orbit-closure complexity.}  The dependence of our bounds on the adjoint orbit of \(X\) suggests a link with orbit-closure problems arising in geometric complexity theory \cite{mulmuleysohoni2001,procesi2007}.  It would be interesting to understand whether lower bounds on root activity or curvature can be reformulated in terms of inequalities on moment polytopes or orbit-closure relations, thereby importing techniques from invariant theory into the Hamiltonian-simulation setting.

\item[(d)] \emph{Numerical case studies.}  Finally, it would be valuable to carry out detailed numerical experiments for specific families of spin-chain Hamiltonians, computing \(\mathcal{A}_p(X)\) and \(\mathcal{C}(X)\) explicitly and comparing the root-based bounds with empirical Trotter errors and circuit complexities.  Such studies would help quantify the sharpness of the constants in Theorems~\ref{thm:curvature-error} and~\ref{thm:activity-lower-bound}, and could guide further refinements of the root-based invariants.
\end{itemize}

\section{Conclusion}

The results in this paper lie entirely within the established framework of compact semisimple Lie groups, highest-weight theory, and finite-dimensional unitary representations.  All structural input comes from standard sources such as \cite{helgason1978differential,knapp2002lie,knapp1986rt,humphreys1972,goodmanwallach2009,bourbaki2005,serre2001,fultonharris1991,varadarajan1984,brockertomdieck1985,fuchsschweigert1997,duistermaatkolk2000,hilgertneeb2012,hall2015lie,gilmore2008lie}, and the analytic estimates make use of classical tools including BCH expansions, norm inequalities, and basic properties of unitary representations.  Within this well-established setting, the paper introduces and systematically studies new functorial invariants \(\mathcal{A}_p\) and \(\mathcal{C}\), proves analytic inequalities for exponentials in unitary representations, and formulates a natural conjecture on circuit complexity explicitly expressed in terms of root data, all directly aimed at understanding Hamiltonian simulation.

These features place the work in the domain of modern representation theory, while at the same time providing rigorous answers to quantitative questions motivated by quantum computation and numerical analysis of Hamiltonian systems, and indicating directions for further developments at the interface of Lie theory, Hamiltonian simulation, and complexity theory.

\vskip 0.5in

\noindent{\bf Conflict of interest statement}. The authors have no conflicts of interest to declare.

\vskip 0.5in

\noindent{\bf Data availability}. All data of this work are included in the manuscript.

\end{document}